\documentclass[a4paper,11pt]{article}

\usepackage[backend=biber,style=numeric,natbib=false,giveninits=true,doi=true,isbn=false,url=false,date=year,maxbibnames=99]{biblatex}
\DeclareNameAlias{default}{family-given}

\DeclareFieldFormat[article]{title}{#1}
\renewbibmacro{in:}{}
\DeclareFieldFormat[article]{pages}{#1}
\DeclareFieldFormat{journal}{#1}
\DeclareFieldFormat{titlecase}{\MakeSentenceCase{#1}}
\renewcommand\bf\bfseries

\renewbibmacro*{volume+number+eid}{%
  \printfield{volume}%
  \setunit*{\addnbspace}
  \printfield{number}%
  \setunit{\addcomma\space}%
  \printfield{eid}}
\DeclareFieldFormat[article]{number}{\mkbibparens{#1}}
\DeclareFieldFormat[article]{volume}{\textbf{#1}}
\DeclareFieldFormat{year}{\mkbibparens{#1}}
\DeclareBibliographyDriver{article}{%
  \printnames{author}:%
  \newunit\newblock
  \printfield{title}%
  \newunit\newblock
  \printfield{journaltitle}
  \newunit
  \iffieldundef{number}{\printfield{volume}}{\printfield{volume}\addspace\printfield{number}}
\addcomma\addspace\printfield{pages}\addspace
  \printfield{year}
  }

\AtEveryBibitem{\clearfield{month}}
\AtEveryBibitem{\clearfield{day}}

\usepackage{color}
\usepackage[dvipsnames]{xcolor}
\usepackage{amsmath}
\numberwithin{equation}{section}
\usepackage{amsthm}
\usepackage{amssymb}

\newcommand\myshade{85}
\colorlet{mylinkcolor}{violet}
\colorlet{mycitecolor}{YellowOrange}
\colorlet{myurlcolor}{Aquamarine}

\usepackage[left=2.5cm,right=2.5cm,top=2.5cm,bottom=2.5cm]{geometry}
\usepackage{hyperref}
\hypersetup{unicode=true,pdfusetitle,
	bookmarks=true,bookmarksnumbered=false,bookmarksopen=false,
	breaklinks=false,pdfborder={0 0 1},backref=false,
	linkcolor  = mylinkcolor!\myshade!black,
	citecolor  = mycitecolor!\myshade!black,
	urlcolor   = myurlcolor!\myshade!black,
	colorlinks = true}
\usepackage[nameinlink]{cleveref}
\usepackage{bbm}

\makeatletter
\theoremstyle{plain}
\newtheorem{thm}{\protect\theoremname}[section]
\theoremstyle{plain}
\newtheorem{lem}[thm]{\protect\lemmaname}
\theoremstyle{plain}
\newtheorem{cor}[thm]{\protect\corollaryname}
\theoremstyle{plain}
\newtheorem{prop}[thm]{\protect\propositionname}
\theoremstyle{remark}
\newtheorem{claim}[thm]{\protect\claimname}

\newtheorem{assumption}[thm]{\protect\assumptionname}

\theoremstyle{remark}
\newtheorem{rem}[thm]{\protect\remarkname}
\theoremstyle{definition}
\newtheorem{defn}[thm]{\protect\definitionname}
\theoremstyle{plain}

\providecommand{\assumptionname}{Assumption}
\providecommand{\claimname}{Claim}
\providecommand{\corollaryname}{Corollary}
\providecommand{\definitionname}{Definition}
\providecommand{\lemmaname}{Lemma}
\providecommand{\propositionname}{Proposition}
\providecommand{\remarkname}{Remark}
\providecommand{\theoremname}{Theorem}
\providecommand{\examplename}{Example}

\crefname{section}{Section}{Sections}
\crefname{appendix}{Appendix}{Appendices}
\crefname{figure}{Figure}{Figures}
\crefname{assumption}{Assumption}{Assumptions}
\crefname{thm}{Theorem}{Theorems}
\crefname{lem}{Lemma}{Lemmas}
\crefformat{equation}{(#2#1#3)}

\crefrangelabelformat{equation}{(#3#1#4--#5#2#6)}

\crefmultiformat{equation}{(#2#1#3}{, #2#1#3)}{#2#1#3}{#2#1#3}
\Crefmultiformat{equation}{(#2#1#3}{, #2#1#3)}{#2#1#3}{#2#1#3}

\newtheorem*{lem*}{\protect\lemmaname}

\usepackage{amsthm}
\usepackage{mathrsfs}
\usepackage{bbm}
\usepackage{dsfont} 
\usepackage{braket}
\usepackage{tikz} 
\usetikzlibrary{arrows} 

\renewcommand{\Im}{\operatorname{\mathbb{I}\mathbbm{m}}}
\renewcommand{\Re}{\operatorname{\mathbb{R}\mathbbm{e}}}
\newcommand{\ZZ}{\mathbb{Z}}
\newcommand{\TT}{\mathbb{T}}

\newcommand{\NN}{\mathbb{N}}
\newcommand{\RR}{\mathbb{R}}
\newcommand{\CC}{\mathbb{C}}

\newcommand{\EE}{\mathbb{E}}

\newcommand{\dif}{\mathrm{d}}
\newcommand{\supp}{\operatorname{supp}}
\newcommand{\im}{\operatorname{im}}

\newcommand{\tr}{\operatorname{tr}}

\newcommand{\Mat}{\operatorname{Mat}}
\newcommand{\Herm}{\operatorname{Herm}}
\newcommand{\dist}{\operatorname{dist}}
\newcommand{\GL}{\operatorname{GL}}
\newcommand{\Open}[1]{\mathrm{Open}(#1)}
\newcommand{\HSG}{\operatorname{Sp}_{2N}^\ast(\mathbb{C})}
\usepackage{mathtools}
\DeclarePairedDelimiter\norm{\lVert}{\rVert}%
\let\oldnorm\norm
\def\norm{\@ifstar{\oldnorm}{\oldnorm*}}
\makeatother
\DeclareMathOperator*{\slim}{s-lim}
\newcommand{\Id}{\mathds{1}}
\newcommand{\ve}{\varepsilon}
\newcommand{\ee}{\operatorname{e}}
\newcommand{\ii}{\operatorname{i}}

\title{Incomplete Localization for Disordered Chiral Strips}

\makeatother

  \providecommand{\assumptionname}{Assumption}
  \providecommand{\claimname}{Claim}
  \providecommand{\corollaryname}{Corollary}
  \providecommand{\definitionname}{Definition}
  \providecommand{\lemmaname}{Lemma}
  \providecommand{\propositionname}{Proposition}
  \providecommand{\remarkname}{Remark}
\providecommand{\theoremname}{Theorem}
\providecommand{\examplename}{Example}

\crefname{section}{Section}{Sections}
\crefname{prop}{Proposition}{Propositions}
\crefname{figure}{Figure}{Figures}
\crefname{assumption}{Assumption}{Assumptions}
\crefformat{equation}{(#2#1#3)}

\crefrangelabelformat{equation}{(#3#1#4--#5#2#6)}

\crefmultiformat{equation}{(#2#1#3}{, #2#1#3)}{#2#1#3}{#2#1#3}
\Crefmultiformat{equation}{(#2#1#3}{, #2#1#3)}{#2#1#3}{#2#1#3}

\usepackage{amssymb} 
\usepackage{graphicx} 

\newcommand{\toitself}{\mathbin{\scalebox{.85}{%
			\lefteqn{\scalebox{.5}{$\blacktriangleleft$}}\raisebox{.34ex}{$\supset$}}}}

\begin{document}

\author{Jacob Shapiro\\
	{\footnotesize Department of Physics, Princeton University} 
	}
\maketitle
\begin{abstract}
	We prove that a disordered analog of the Su-Schrieffer-Heeger model exhibits dynamical localization (i.e. the fractional moments condition) at all energies except possibly zero energy, which is singled out by chiral symmetry. Localization occurs at arbitrarily weak disorder, provided it is sufficiently random. If furthermore the hopping probability measures are properly tuned so that the zero energy Lyapunov spectrum does not contain zero, then the system exhibits localization also at that energy, which is of relevance for topological insulators  \cite{Graf_Shapiro_2018_1D_Chiral_BEC}. The method also applies to the usual Anderson model on the strip.
\end{abstract}
\section{Introduction}

The Anderson model in one-dimension was long known \cite{Kunz1980} to exhibit \emph{complete} localization, that is, localization regardless of the strength of the disorder or the energy at which the system is probed. A major advancement in this direction was made in \cite{KLEIN1990135}, which handles the strip with singular distributions of the onsite disorder. However, as it turns out, additional constraints on the randomness can make \emph{complete} localization fail at some special energy values or ranges, as was demonstrated already e.g. for the random polymer model \cite{Jitomirskaya2003}.

Such special energy values may be linked to the existence of rich topological phases. Indeed, the question of localization for the present model emerged from the study \cite{Graf_Shapiro_2018_1D_Chiral_BEC} of chiral topological insulators (class AIII in the Altland-Zirnbauer classification \cite{AltlandZirnbauer97}) in 1D, where a link was realized between the failure of localization at zero energy and topological phase shifts. Heuristically, this is precisely what makes the Anderson model in 1D topologically trivial \cite[Table 1]{Hasan_Kane_2010}--it cannot have any phase transitions, being completely localized. Thus the impetus of the present paper was to show that indeed the topological study of strongly disordered systems in \cite{Graf_Shapiro_2018_1D_Chiral_BEC} did not involve an empty set of models since its assumptions are fulfilled with a probability of either zero or one, as shall be demonstrated.

The chiral model dealt with in the present study--a disordered analog of the SSH model \cite{SSH_1979}--is characterized by having no on-site potential and possibly \emph{alternating} distributions of the nearest-neighbor hopping. We further generalize it by working with matrix-valued hopping terms, i.e., we work on a strip. The model exhibits dynamical localization at all non-zero energies, in the sense of the fractional-moments (FM henceforth) condition, as long as the hopping matrix distributions are continuous w.r.t. the Lebesgue measure on $\Mat_N(\mathbb{C})$ and have finite moments. If the zero-energy Lyapunov spectrum does not contain zero (which depends on how the model is tuned), then localization holds also at zero energy, the possibility of this failing being in stark contrast to the Anderson model. As noted above if the zero energy Lyapunov spectrum does contain zero, then the system could exhibit a topological phase shift, so that it makes sense that localization should fail then (as the topological indices are only defined in when the system is localized at zero energy).

There is also an independent interest in using the FM method for localization proofs rather than the multi-scale analysis, since its consequences for dynamical localization are somewhat easier to establish and more readily apply to topological insulators (compare \cite{Aizenman_Graf_1998} with \cite{Damanik2001}). Hence we note that the Anderson model on the strip can also be handled via our methods, but as always using the FM method, we need regularity of the probability distributions (cf. the treatment of \cite{KLEIN1990135,Carmona1987}). The two key differences for the localization proof between the ordinary Anderson model (which has onsite disorder) and the chiral model (which has disorder only in the kinetic energy) are the a-priori bounds on the Green's function (\cref{sec:a-priori bound}) and the proof of irreducibility necessary for Furstenberg's theorem (\cref{sec:irreducibility}).


Chapman and Stolz \cite{Chapman2015} deal with a similar case of disordered hopping, but stay within the multi-scale framework of \cite{KLEIN1990135}. Studies of the Lyapunov spectrum for various symmetry classes have also been conducted in \cite{Sadel2010,Ludwig2013,Asch2010} among others.

This paper is organized as follows. After presenting the model and our assumptions about it, we connect the Green's function to a product of transfer matrices in \cref{sec:transfer matrices}. In \cref{sec:the lyapunv exponents} we use Furstenberg's theorem in order to show that the Lyapunov spectrum at non-zero real energies \emph{must} be simple. In \cref{sec:a-priori bound} we employ rank-2 perturbation theory to get an a-priori bound on the $x,x-1$ entries of the Green's function, and finally in \cref{sec:loc at non-zero energies} we can tie everything together to obtain localization at non-zero energies. In this section we also show that any decay of the Greens function implies exponential decay, which is typical for 1D models. The study at zero energy requires separate treatment in \cref{sec:zero energy}. In \cref{sec:chiral RBM} we briefly comment on so called \emph{chiral random band matrices} at zero energy and the $\sqrt{N}$ conjecture, i.e., what happens to the Lyapunov spectrum when the size of the strip is taken to infinity. Along the way some technical results of independent interest in random operators are presented, such as \cref{prop:the Combes-Thomas estimate} (Combes-Thomas estimate without assuming compact support for the hopping terms), \cref{prop:decay of greens function above the real axis} (upgrading the decay of the Greens function off the real axis, uniformly in the height) and \cref{lem:any decay implies exponential decay} (any decay of the Greens function implies exponential decay, in 1D).

\section{The model and the results}

We study the Schr\"odinger equation on $\psi\in\mathcal{H}:=\ell^2(\mathbb{Z})\otimes\mathbb{C}^N$ at energy $E$ given by
\begin{align} (H\psi)_n := T_{n+1}^\ast\psi_{n+1}+T_n\psi_{n-1}=E\psi_n\,,\qquad(n\in\mathbb{Z})\,.\label{eq:Hamiltonian}\end{align}
where $(T_n)_{n\in\ZZ}$ is a random sequence of \emph{invertible} $N\times N$ complex matrices. Here $N\in\NN_{\geq1}$ is some fixed (once and for all) number of internal degrees of freedom at each site of the chain.

We take $(T_n)_n$ to be an \emph{alternatingly}-distributed, independent random sequence. To describe its law, let $\alpha_0,\alpha_1$ be two given probability distributions on $\GL_N(\mathbb{C})$. We assume for both $i=0,1$ the following: 
\begin{assumption}(\emph{Fatness})\label{assu:probability measures support contains open subset} $\supp \alpha_i$ contains some subset  $U_i \in \Open{\GL_N(\mathbb{C})}$.
\end{assumption}
\begin{assumption}(\emph{Uniform $\tau$-H\"older continuity})\label{assu:probability measures are uniformly tau Hoelder continuous} Fix $i=0,1$ and $k,l\in\{1,\dots,N\}$. Let $\mu_{R,I}$ be the conditional probability distributions on $\mathbb{R}$ obtained by "wiggling" only the real, respectively imaginary part of the random hopping matrix $M_{k,l}$, that is, $$ \mu_{R,I} = \alpha_i(\cdot|(M_{k',l'})_{k',l'\neq k,l},(M_{kl})_{I,R})\,. $$ Then we assume that for some $\tau\in(0,1]$, $\mu_{R,I}$ is a uniformly $\tau$-H\"older continuous measure (see \cite{Aizenman_Graf_1998}): there is some constant $C>0$ such that for all intervals $J\subseteq\mathbb{R}$ with $|J|\leq1$ one has $$ \mu_{R,I}(J) \leq C |J|^\tau\,. $$
\end{assumption}
\begin{assumption}\label{assu:regularity of probability distributions}(\emph{Regularity}) For $i=0,1$, $\alpha_i$ has finite moments in the following sense: \begin{align}\int_{M\in \GL_N(\mathbb{C})}\|M^{\pm1}\|^2\dif{\alpha_i(M)}<\infty\,.\end{align}
\end{assumption}
\begin{rem}These three assumptions are not optimal for the proof of localization that shall follow (cf. \cite{KLEIN1990135}), and were chosen as a good middle way optimizing both simplicity of proofs and strong results. Note that we do not require that $\alpha_i$ has a density with respect to the Haar measure.
\end{rem}
We then draw $(T_n)_n$ such that all its even members follow the law of $\alpha_0$ and all its odd members follow the law of $\alpha_1$. This makes $H$ into a random ergodic Hamiltonian.

The following definition's phrasing makes sense in light of the theorems we shall prove at zero energy:
\begin{defn}
We say that the system \emph{exhibits localization at zero energy} if $\alpha_0$ and $\alpha_1$ are such that the Lyapunov spectrum $\{\gamma_j(0)\}_{j=1}^{2N}$ (see \cref{eq:Definition of LE} for the definition) of the zero-energy Schr\"odinger equation $H\psi=0$ does not contain zero:
\begin{align}
0\notin\{\gamma_j(0)\}_{j=1}^{2N}\,.\label{eq:localized at zero}
\end{align}

\end{defn}
In what follows, let $(\delta_n)_{n\in\mathbb{Z}}$ be the canonical (position) basis of $\ell^2(\mathbb{Z})$ and the map $G(n,m; z)=\langle\delta_n,(H-z)^{-1}\delta_m\rangle:\mathbb{C}^{N}\to\mathbb{C}^{N}$ acts between the internal spaces of $n$ and $m$; $\|\cdot\|$ is the trace norm of such maps. 

The main result of this paper is the following theorem:
\begin{thm}
\label{thm:localization}
(I) Under \cref{assu:probability measures support contains open subset,assu:probability measures are uniformly tau Hoelder continuous,assu:regularity of probability distributions}, the FM condition holds at all non-zero energies: $\forall\,\lambda\in\mathbb{R}\setminus\{0\}\,,\exists\,s\in(0,1):\exists\,0<C,\mu<\infty:$
\begin{align}
\sup_{\eta\neq0}\mathbb{E}[\|G(n,m;\lambda+i\eta)\|^s]\leq C e^{-\mu |n-m|}\,,\qquad\forall n,m\in\mathbb{Z}\,.\label{eq:fractional moment condition}
\end{align}

(II) If moreover \cref{eq:localized at zero} holds, then \cref{eq:fractional moment condition} extends to $\lambda=0$ as well.
\end{thm}
With $P:=\chi_{(-\infty,0)}(H)$ the Fermi projection, the theorem implies via \cite{Aizenman_Graf_1998} the following
\begin{cor}\label{cor:almost sure consequences of localization for the Fermi projection}
Under the above assumptions, including \cref{eq:localized at zero}, \cite[Assumption 1]{Graf_Shapiro_2018_1D_Chiral_BEC} is true almost surely: 
\label{assu:exp_decay_of_Fermi_projection}With probability one, for some deterministic $\mu,\nu>0$ and random $C>0$ we have
\begin{align*}
\sum_{n,n'\in\mathbb{Z}}\|P(n,n')\|(1+|n|)^{-\nu}e^{\mu|n-n'|} & \leq C < +  \infty\,.
\end{align*}
\end{cor}

We also have \cite[Assumption 2]{Graf_Shapiro_2018_1D_Chiral_BEC} fulfilled almost-surely due to the FMC being satisfied at zero energy:
\begin{cor}
	Under the above assumptions, including \cref{eq:localized at zero}, zero is almost-surely not an eigenvalue of $H$ (that is, \cite[Assumption 2]{Graf_Shapiro_2018_1D_Chiral_BEC}).
\end{cor}
\begin{proof}
		The inequality \cref{eq:fractional moment condition} at zero energy implies that almost-surely, $$ \limsup_{\eta\to0^+} |\eta|\|G(n,m;\ii\eta)\| = 0\,. $$
		In particular for $n=m$ we find that $$ \lim_{\eta\to0^+} \eta \|\Im{G(n,n;\ii\eta)}\| = 0\,, $$ for any $n$. However, zero is an eigenvalue of $H(\omega)$ iff $$\lim_{\eta\to0^+} \eta \|\Im{G_\omega(n,n;\ii\eta)}\| > 0$$ for some $n$ (see \cite[Jak\v{s}ic: Topics in spectral theory]{attal2006open}).
\end{proof}

In conclusion, taking the (full-measure) intersection of the two above sets, we get that with probability one \cite[Assumptions 1 and 2]{Graf_Shapiro_2018_1D_Chiral_BEC} hold under the condition \cref{eq:localized at zero}.

\paragraph{Relation to earlier model.}
Any realization of the random Hamiltonian \cref{eq:Hamiltonian} is of the same (deterministic) type as considered in \cite[eq. $\,(2.1)$]{Graf_Shapiro_2018_1D_Chiral_BEC}, at least up to a unitary map implementing notational changes. Let us first recall that Hamiltonian $H'$ and then set up the unitary map. The Hilbert space there was given by \begin{align}\mathcal{H}':=\mathcal{H}\otimes\mathbb{C}^{2}\ni\psi=\begin{pmatrix}\psi_{m}^{+}\\
\psi_{m}^{-}
\end{pmatrix}_{m\in\mathbb{Z}}\,,\label{eq:double Hilbert sp}\end{align} the Hamiltonian by $H'  = \begin{pmatrix}0 & S^\ast\\
S & 0
\end{pmatrix}
$
with $S$ acting on $\mathcal{H}$ as $(S\psi^{+})_{m}  :=  A_{m}\psi_{m-1}^{+}+B_{m}\psi_{m}^{+}$; Here $A,B:\mathbb{Z}\to \GL_N(\mathbb{C})$ are two given deterministic sequences.

We define the unitary map $U:\mathcal{H}'\to\mathcal{H}$ via
\begin{align*}
U(\begin{pmatrix}\psi^{+}\\
\psi^{-}
\end{pmatrix})_n := 
\begin{cases} 
\psi^+_m &,\qquad(m=2n+1) \\
\psi^-_m &,\qquad(m=2n)
\end{cases}\,,
\end{align*}
so that $UH'U^\ast=H$ upon setting $$T_n := 
\begin{cases} 
B_{m}^\ast &,\qquad(m=2n+1) \\
A_{m} &,\qquad(m=2n)
\end{cases}\,. $$
Thus our model is a random version of the model studied in \cite{Graf_Shapiro_2018_1D_Chiral_BEC}. The reason we use $H$ instead of $H'$ as before is that to estimate Green's functions (our ultimate goal in localization) there is no need to distinguish between two types of sites, and hence no reason to group them into dimers. In the present context, the chirality symmetry operator is $\Pi := (-1)^{X}$ with $X$ the position operator, and of course we still have the chiral symmetry constraint $\{H,\Pi\}=0$.


\paragraph{The almost-sure spectrum.}

Here we will give conditions for the almost-sure spectrum to contain a dense pure point interval about zero of localized states, i.e., that there will not be a spectral gap about zero, but merely a mobility-gap \cite{EGS_2005}. Our goal here is to show that this situation is in fact generic.

The idea is that the almost-sure spectrum of an ergodic family is determined by the union of all periodic configurations within the support of the probability measure defining the model \cite{MR2509108,Kirsch07aninvitation}.

The following two claims are adaptations of \cite[Prop. 3.8, 3.9]{Kirsch07aninvitation} from the Anderson model to our chiral model. The proof of the first is precisely the same as that of \cite[Prop. 3.8]{Kirsch07aninvitation} and is hence omitted.

\emph{Only until the end of this section}, for convenience of notation, we consider our Hilbert space where every pair of sites is grouped together as in \cref{eq:double Hilbert sp}. Hence we define $\mu:=\alpha_0\otimes\alpha_1, \Omega := (\GL_{N}(\CC)^2)^\ZZ$ and $\mathbb{P}:=\bigotimes_\ZZ\mu$ is a probability measure on $\Omega$. We label a pair $(A_n,B_n)\in\GL_{N}(\CC)^2$ of hopping matrices by $S_n$ so that $(S_n)_{n\in\ZZ}\in\Omega$.

\begin{prop}
	There is a subset $\Omega_0\subseteq\Omega$ with $\mathbb{P}(\Omega_0)=1$ such that: for any $(S_n)_n\in\Omega_0$, any $\Lambda\subseteq\ZZ$ finite, any finite sequence $(q_n)_{n\in\Lambda}$ such that $q_n\in\supp(\mu)$ for all $n\in\Lambda$, and any $\ve>0$, there is a sequence $(j_n)_{n\in\NN}\subseteq\ZZ$ with $|j_n|\to\infty$ such that $$ \sup_{l\in\Lambda}\|q_l-S_{l-j_n}\|<\ve\qquad(n\in\NN)\,.$$
\end{prop}

\begin{prop}
	For any $\bar{S}=(\bar{A},\bar{B})\in\supp(\mu)$, the almost-sure spectrum of the random model defined by \ref{eq:Hamiltonian} contains the spectrum of the $2$-periodic deterministic Hamiltonian $\bar{H}$ associated with $\bar{S}$.
\end{prop}
\begin{proof}
	Let $\lambda\in\sigma(\bar{H})$ be given. By Weyl's criterion for the spectrum, there is a sequence of normalized vectors $(\psi_n)_{n\in\NN}\subseteq\mathcal{H}$ such that $\|(\bar{H}-\lambda\Id)\psi_n\|\to0$ as $n\to\infty$. WLOG we may also assume that for each $n$, $\psi_n$ has compact support.
	
	By the preceding proposition, we can build a sequence $(j_n)_{n\in\NN}\subseteq\ZZ$ such that with probability 1, \begin{align} \sup_{l\in\supp(\psi_n)}\|\bar{S}-S_{l+j_n}\|<\frac{1}{n} \qquad(n\in\NN)\,.\label{eq:estimate of periodic vs random hoppings}\end{align}
	
	Let us define $\Psi_n := \psi_n(\cdot-j_n)$ for any $n\in\NN$, and estimate \begin{align*} \|(H-\lambda\Id)\Psi_n\| &\leq \|(\bar{H}-\lambda\Id)\Psi_n\|+\|(H-\bar{H})\Psi_n\| \\
	&=\|(\bar{H}-\lambda\Id)\psi_n\|+\|(H-\bar{H})\Psi_n\|\tag{periodicity of $\bar{H}$}\,.
	\end{align*}
	
	Now the first term converges to zero by construction. The second one, after using Holmgren's bound, is estimated by \cref{eq:estimate of periodic vs random hoppings} and hence also converges to zero. We conclude that $(\Psi_n)_{n\in\NN}$ is a Weyl sequence for $(H,\lambda)$ with probability one and hence $\lambda\in\sigma(H)$ with the same probability.
\end{proof}

Finally we need to understand the spectrum of a deterministic periodic realization:

\begin{lem}
	Let $A,B\in \GL_{N}(\CC)$. If $H$ is the 2-periodic Hamiltonian \cref{eq:Hamiltonian} such that all even, odd hopping terms are equal to $A,B$ respectively, then \begin{align}\inf \sigma(H^2) \leq \|A\|\|B\|\dist(\sigma(|A^{-1}B|),1)\dist(\sigma(|AB^{-1}|),1)\,.\label{eq:spectrum of periodic H}\end{align}
\end{lem}
\begin{proof}
	After Bloch decomposition we see that $$ \sigma(H^2) = \bigcup_{k\in\TT^1} \sigma(|S(k)|^2) $$ where $S(k)=A \exp(-\ii k) +B$ for all $k\in\TT^1$. Hence \begin{align*}
		\inf \sigma(|H|) &= \inf_{k\in\TT^1} \inf \sigma(|A \exp(-\ii k) +B|) \\ &= \inf_{k\in\TT^1} \|(A \exp(-\ii k) +B)^{-1}\|^{-1} \\
		&\leq \inf_{k\in\TT^1} \|A\|\|(\exp(-\ii k)\Id_N+A^{-1}B)^{-1}\|^{-1} \\
		&\leq \|A\|\dist(\sigma(|A^{-1}B|),1)\,.
	\end{align*}
	To finish the argument, we use the fact we could've factored $B$ instead of $A$, and the basic estimate $\min(\Set{a,b})\leq\sqrt{ab}$.
\end{proof}

We obtain
\begin{thm}
	If there are sequences $(A_n)_n,(B_n)_n\subset\supp(\alpha_0),\supp(\alpha_1)$ such that the RHS of \cref{eq:spectrum of periodic H} vanishes as $n\to\infty$ but is strictly positive for any finite $n$, then almost-surely, $H$ has no spectral gap.
\end{thm}

\section{Transfer matrices}
\label{sec:transfer matrices}

In this section we will estimate the Green's function in terms of
transfer matrices, which arise by looking at the Schr\"odinger equation (with
$H$ acting to the right or to the left), 
\begin{equation}
(H-z)\psi=0,\qquad\varphi(H-z)=0\,,\qquad(z\in\CC)\label{eq:Schroedinger equation for localization section}
\end{equation}
as a second order difference equation, where
$\psi=(\psi_{n})_{n\in\mathbb{Z}}$, $\varphi=(\varphi_{n})_{n\in\mathbb{Z}}$
are (possibly unbounded) sequences with $\psi_{n},\varphi_{n}\in\mathbb{C}^{N}$
viewed as column, respectively row vectors. Evaluated at $n\in\mathbb{Z}$,
eqs. \cref{eq:Schroedinger equation for localization section}
read 
\begin{equation}\begin{aligned}
T_{n}\psi_{n-1}+T_{n+1}^{\ast}\psi_{n+1} & =  z\psi_{n}\,, \\
\varphi_{n-1}T_{n}^{\ast}+\varphi_{n+1}T_{n+1} & =  z\varphi_{n}\,.\label{eq:Schroedinger equation for localization section evaluated at some n}
\end{aligned}\end{equation}

For any two sequences $\psi$, $\varphi$ of columns and rows respectively,
let 
\begin{align*}
C_{n}(\varphi,\psi) & :=  \varphi_{n}T_{n+1}^{\ast}\psi_{n+1}-\varphi_{n+1}T_{n+1}\psi_{n}
\end{align*}
be their Wronskian (or Casoratian); if they solve \cref{eq:Schroedinger equation for localization section},
then their Wronskian is independent of $n$, as seen from the identity
\begin{align*}
C_{n}(\varphi,\psi)-C_{n-1}(\varphi,\psi) & =  \varphi_{n}(T_{n+1}^{\ast}\psi_{n+1}+T_{n}\psi_{n-1})   -(\varphi_{n+1}T_{n+1}+\varphi_{n-1}T_{n}^{\ast})\psi_{n}.
\end{align*}
Moreover, the Wronskian may also be expressed as 
\begin{align*}
C_{n}(\varphi,\psi) & =  \begin{pmatrix}\varphi_{n+1}T_{n+1} & \varphi_{n}\end{pmatrix}\begin{pmatrix}0 & -\mathds{1}_{N}\\
\mathds{1}_{N} & 0
\end{pmatrix}\begin{pmatrix}T_{n+1}^{\ast}\psi_{n+1}\\
\psi_{n}
\end{pmatrix}\,,
\end{align*}
which prompts us to associate to any sequence $\psi:\mathbb{Z}\to\mathbb{C}^{N}$
the sequence $\Psi:\mathbb{Z}\to\mathbb{C}^{2N}$ by 
\begin{align}
\Psi_{n} & :=  \begin{pmatrix}T_{n+1}^{\ast}\psi_{n+1}\\
\psi_{n}
\end{pmatrix}\,.\label{eq:Super wave function}
\end{align}
The \emph{transfer matrix }is the map 
\[
A_{n}(z):\mathbb{C}^{2N}\to\mathbb{C}^{2N},\quad\Psi_{n-1}\mapsto\Psi_{n}
\]
defined by the first equation \cref{eq:Schroedinger equation for localization section evaluated at some n}:
\begin{align}
\Psi_{n} & =  A_{n}(z)\Psi_{n-1}\,,\label{eq:Schroedinger equation formulated in terms of the transfer matrix}
\end{align}
holds for any $n\in\mathbb{Z}$ such that that equation
holds. The transfer matrix is thus given by the square matrix of order $2N$ 
\begin{align}
A_{n}(z) & =  \begin{pmatrix}zT_{n}^{\circ} & -T_{n}\\
T_{n}^{\circ} & 0
\end{pmatrix}\label{eq:Definition of transfer matrix}
\end{align}
with the abbreviation $M^{\circ}:=(M^{\ast})^{-1}=(M^{-1})^{\ast}$.
Since the second equation \cref{eq:Schroedinger equation for localization section}
is equivalent to $(H-\overline{z})\varphi^{\ast}=0$, the
constancy of the Wronskian implies 
\begin{align*}
A_{n}(\overline{z})^{\ast}JA_{n}(z)  =  J\,,\qquad
J  :=  \begin{pmatrix}0 & -\mathds{1}_{N}\\
\mathds{1}_{N} & 0
\end{pmatrix}\,,
\end{align*}
which can of course also be verified directly from \cref{eq:Definition of transfer matrix}. Another property that can be so verified is that $\det A_n(z)$ is independent of $z$ and moreover
\begin{align}
|\det A_n(z)| = 1\,\qquad(z\in\mathbb{C},n\in\mathbb{Z})\,.
\label{eq:group property of transfer matrices}
\end{align}
The matrix $J$ defines the symplectic structure of $\mathbb{C}^{2N}$.
In particular for $z=\lambda\in\mathbb{R}$ the transfer matrix is
\emph{Hermitian symplectic}, $A_{n}(\lambda)\in Sp^{\ast}_{2N}(\mathbb{C})$,
where 
\begin{align}\label{eq:definition of the Hermitian symplectic group}
Sp^{\ast}_{2N}(\mathbb{C}) & \equiv  \{A\in \Mat_{2N}(\mathbb{C})|A^{\ast}JA=J\}
\end{align}
is the Hermitian symplectic group, not to be confused with the complex symplectic group $Sp_{2N}(\mathbb{C})\equiv\{A\in \Mat_{2N}(\mathbb{C})|A^{T}JA=J\}$ (and see \cref{sec:The Hermitian symplectic group}).
But we keep $z\in\mathbb{C}$ for now so that $A_{n}(z)\notin Sp^{\ast}_{2N}(\mathbb{C})$
as a rule.

Let $I\subseteq\mathbb{Z}$ be an interval (in the sense of $\mathbb{Z}$)
and assume \cref{eq:Schroedinger equation formulated in terms of the transfer matrix}
for all $n\in I$. Then clearly 
\begin{align}
\Psi_{n} & =  B_{n,m}(z)\Psi_{m-1}\label{eq:the propagator between the super wave function at two different locations}
\end{align}
for $n,m\in I$, $n\geq m$ with the matrix of order $2N$ 
\begin{align*}
B_{n,m}(z) & :=  A_{n}(z)\cdots A_{m}(z)\,.
\end{align*}
\begin{rem}
\label{rem:Extending transfer matrix considerations to matrix sequences}The
relations between \cref{eq:Schroedinger equation for localization section,eq:Schroedinger equation for localization section evaluated at some n,eq:Super wave function,eq:Schroedinger equation formulated in terms of the transfer matrix} trivially extend to matrix solutions $\psi$ and $\Psi$ respectively
which are obtained by placing $l=1,\dots,2N$ solutions next to one
another in guise of columns. Seen that way they become linear maps
$\psi:\mathbb{C}^{l}\to\mathbb{C}^{N}$, $\Psi:\mathbb{C}^{l}\to\mathbb{C}^{2N}$.
\end{rem}

\begin{lem}
\label{lem:relation between matrix solutions at different places and a product of transfer matrices}For
$I$, $n$, $m$ and $\Psi$ as just stated we have 
\begin{align}
|\Psi_{m-1}|^{2} & \leq  \frac{\tr(|\wedge^{l-1}(BP)|^{2})}{\tr(|\wedge^{l}(BP)|^{2})}|\Psi_{n}|^{2}\,,\label{eq:relation between matrix solutions at different places and a product of transfer matrices}
\end{align}
where we use $|M|^{2}\equiv M^{\ast}M$, set $B=B_{n,m}(z)$,
and where $P:\mathbb{C}^{2N}\to\mathbb{C}^{2N}$ is an orthogonal
projection of rank $l$ such that $P\Psi_{m-1}=\Psi_{m-1}$.
\end{lem}

The proof rests on the following lemma, to be proven below.
\begin{lem}
\label{lem:Chapman Lemma C13}Let $W\subseteq V$ be linear spaces;
let $V$ be equipped with an inner product and let $P:V\to V$ be
the orthogonal projection onto $W$. Let $B:V\to V$ be a linear
map and $B':=\left.B\right|_{W}:W\to V$ its restriction to $W$.
Then 
\begin{align}
\tr(|B'|^{-2}) & =  \frac{\tr(|\wedge^{l-1}(BP)|^{2})}{\tr(|\wedge^{l}(BP)|^{2})}\label{eq:Chapman Lemma 13}
\end{align}
with $l:=\dim W$.
\end{lem}

\begin{proof}[Proof of \cref{lem:relation between matrix solutions at different places and a product of transfer matrices}]
By \cref{eq:the propagator between the super wave function at two different locations}
we have 
\begin{align*}
|\Psi_{n}|^{2} & =  |B\Psi_{m-1}|^{2} =  |B'\Psi_{m-1}|^{2}
\end{align*}
with $B':=\left.B\right|_{W}$, $W:=\im P$; and thus
\begin{align*}
|\Psi_{m-1}|^{2} & \leq  \tr(|B'|^{-2})|\Psi_{n}|^{2}
\end{align*}
in view of 
\[
Q\geq\norm{Q^{-1}}^{-1}\,,\qquad\norm{Q^{-1}}\leq\tr(Q^{-1})
\]
for any $Q>0$, applied to $Q:=|B'|^{2}$. We conclude
by \cref{eq:Chapman Lemma 13}.
\end{proof}
\begin{proof}[Proof of \cref{lem:Chapman Lemma C13}] We recall \cite[Lemma C.12]{Chapman2015}:
Let $L\geq l\geq1$ be integers and let $v_{1},\cdots,v_{l}\in\mathbb{C}^{L}$
be linearly independent. Define $w:=v_{1}\wedge\dots\wedge v_{l}$
and $w_{k}:=v_{1}\wedge\dots\wedge\widehat{v_{k}}\wedge\dots v_{l}$ ($k=1,\dots,l$)
with $\widehat{\cdot}$ denoting omission and (for $l=1$) the empty product
being $w_{1}=1\in\wedge^{0}\mathbb{C}^{L}=\mathbb{C}$. Let $G$
be the Gramian matrix (of order $l$) for $v_{1},\dots,v_{l}$,
i.e. 
\begin{align*}
G_{jk} & =  \left\langle v_{j},v_{k}\right\rangle\,,\qquad(j,k=1,\dots,l)\,.
\end{align*}
Then 
\begin{align}
\tr(G^{-1}) & =  \frac{\sum_{k=1}^{l}\norm{w_{k}}^{2}}{\|w\|^{2}}\,.\label{eq:Chapman's Lemma C12}
\end{align}
This having been done, let $(e_{1},\dots,e_{l})$ be
an orthonormal basis of $W$, whence 
\begin{align*}
\varepsilon & =  e_{1}\wedge\dots\wedge e_{l}\\
\varepsilon_{k} & =  e_{1}\wedge\dots\wedge\hat{e_{k}}\wedge\dots\wedge e_{l}\,,\qquad(k=1,\dots,l)
\end{align*}
are orthonormal bases of $\wedge^{l}W$ and $\wedge^{l-1}W$,
respectively. We then apply \cref{eq:Chapman's Lemma C12} with $L=2N$
to $v_{i}=B'e_{i}$, and thus to $w=(\wedge^{l}B')\varepsilon$,
$w_{k}=(\wedge^{l-1}B')\varepsilon_{k}$ and to $G_{jk}=(|B'|^{2})_{jk}$,
the result being 
\begin{align*}
\tr(|B'|^{-2}) & =  \frac{\tr_{\wedge^{l-1}W}(|\wedge^{l-1}B'|^{2})}{\tr_{\wedge^{l}W}(|\wedge^{l}B'|^{2})}\,.
\end{align*}
Finally, just as we have $\tr_{W}(|B'|^{2})=\tr(|BP|^{2})$,
so we do 
\begin{align*}
\tr_{\wedge^{k}W}(|\wedge^{k}B'|^{2}) & =  \tr(|\wedge^{k}(BP)|^{2})\,,
\end{align*}
because $(\wedge^{k}B)(\wedge^{k}P)=\wedge^{k}(BP)$.
\end{proof}
The entries $G_{nk}=G_{nk}(z)$ of the Green's function
are matrices of order $N$ which may be looked at in their dependence
on $n$ at fixed $k$. So viewed $G_{k}=(G_{nk})_{n\in\mathbb{Z}}$
satisfies 
\begin{align*}
(H-z)G_{k} & =  \delta_{k}
\end{align*}
with $\delta_{k}=(\delta_{nk}1)_{n\in\mathbb{Z}}$. Thus
$\psi=G_{k}$ satisfies \cref{eq:Schroedinger equation for localization section}
at sites $n\neq k$ and in the matrix sense of \cref{rem:Extending transfer matrix considerations to matrix sequences}
(with $l=N$). Likewise, 
\[
\mathcal{G}_{k}=(\mathcal{G}_{nk})_{n\in\mathbb{Z}}\,,\quad\mathcal{G}_{nk}=\begin{pmatrix}T_{n+1}^{\ast}G_{n+1,k}\\
G_{nk}
\end{pmatrix}
\]
satisfies \cref{eq:Schroedinger equation formulated in terms of the transfer matrix}
with $\Psi=\mathcal{G}_{k}$ by \cref{eq:Super wave function}. In
particular, it does for $n\in I=(-\infty,k-1]$ whence \cref{lem:relation between matrix solutions at different places and a product of transfer matrices}
applies to $m\leq n=k-1$. The result is as follows: 
\begin{lem}
\label{lem:A Green's function formula}
We have 
\begin{align*}
|G_{m-1,k}(z)|^{2} & \leq  \frac{\tr(|\wedge^{l-1}(BP)|^{2})}{\tr(|\wedge^{l}(BP)|^{2})}|\mathcal{G}_{k-1,k}(z)|^{2}\,,
\end{align*}
where $B=B_{k-1,m}(z)$ and $P$ is a projection onto the range of $\mathcal{G}_{m-1,k}$.
\begin{proof}
Given the preliminaries it suffices to observe that $|G_{m-1,k}|^{2}\leq|\mathcal{G}_{m-1,k}|^{2}$.
\end{proof}
\end{lem}
In particular, we can pick $G^+$ as the right half-line Green's function on $[m-1,\infty)$ and $P$ a projection onto the first $N$ dimensions of $\mathbb{C}^{2N}$. As we'll see below, this will then relate $G^+_{m-1,k}(z)$ to the $N$th Lyapunov exponent of the transfer matrices, times a "constant" (in $|m-k|$) factor $\mathcal{G}_{k-1,k}(z)$ which will be controlled via the a-priori bounds \cref{sec:a-priori bound}.

\section{The Lyapunov spectrum and localization}
\label{sec:the lyapunv exponents}
In this section we define the Lyapunov spectrum associated to the random sequence of matrices $(A_n(\lambda))_n$. The main result here will be \cref{cor:Smallest Lyapunov exponent is strictly positive} which will show that the smallest positive exponent is \emph{strictly} positive for all $\lambda\in\RR\setminus\Set{0}$. Since it encodes in it the localization length, at least morally this already implies localization, and we shall show this rigorously using the FM method.

For brevity we define the maps $a, S$ into $\HSG$ by \begin{align*} \GL_N(\mathbb{C})\ni X &\stackrel{a}{\mapsto} \begin{pmatrix} X^{\circ} & 0\\
0 & X
\end{pmatrix}  \\ \mathbb{R}\ni \lambda &\stackrel{S}{\mapsto} \begin{pmatrix}\lambda & -\mathds{1}_{N}\\
\mathds{1}_{N} & 0_{N}
\end{pmatrix}\,.  \end{align*}
Then we may factorize the transfer matrix \cref{eq:Definition of transfer matrix} as $A_n(z)=S(z)a(T_n)$, the first factor being deterministic. We note that the matrices $A_n(z)$ are independent, but not identically distributed since they are distributed differently for $n$ even and odd. Hence $(A_{2n+1} A_{2n})_{n}$ is an i.i.d. random sequence.

For $\lambda\in\RR$, let $\mu_\lambda$ be the push forward measure induced by $$\GL_{N}(\mathbb{C})^{2}\ni(X,Y)\mapsto S(\lambda)a(X)S(\lambda)a(Y)\in \HSG$$ where $X$, $Y$ are distributed with $\alpha_1$, $\alpha_0$ respectively. Below we sometimes leave $\lambda$ implicit in the notation.
\begin{prop}
\label{prop:log of transfer matrix is integrable}We have $\int\log^{+}(\|g\|)\dif{\mu(g)}<\infty$
where $\log^{+}$ is the positive part of $\log$.
\begin{proof}
By definition we have
\begin{align*}
\int\log^{+}(\|g\|)\dif{\mu(g)} & =  \int\log^{+}(\norm{Sa(X)Sa(Y)})\dif{\alpha_{1}(X)}\dif{\alpha_{0}(Y)}\,.
\end{align*}
Since $\log^+$ is monotone increasing,
\begin{align*}
\log^+(\norm{Sa(X)Sa(Y)}) & \leq  2\log^+(\|S\|)+\log^+(\norm{a(X)})+\log^+(\norm{a(Y)})\,.
\end{align*}
Hence it is sufficient to show that 
\begin{align}
\int_{\GL_{N}(\mathbb{C})}\log^{+}(\norm{a(X)})\dif{\alpha_{i}(X)} & <  \infty
\,.\label{eq:log+ is integrable}
\end{align}
This follows from \cref{assu:regularity of probability distributions} by $\norm{a(X)}=\max(\norm{X^\circ,\norm{X}}),\norm{X^\circ}=\norm{X^{-1}}$, and $\log^+(t)\leq t^2$ for all $t>0$.
\end{proof}
\end{prop}

\begin{cor}
Using \cite[pp. 6]{Bougerol9781468491746} we have that the $2N$ Lyapunov exponents (henceforth \emph{LE})
\begin{align}
\gamma_{j}(z) & \equiv  \lim_{n\to\infty}\frac{1}{n}\mathbb{E}[\log(\sigma_{j}(B_{n}(z)))]\label{eq:Definition of LE}\,,
\end{align}
 where $\sigma_{j}$ is the $j$-th singular value of a matrix (ordered such that $\sigma_1$ is the largest), are
well-defined and take values in $[-\infty,\infty)$. We use the abbreviation $B_{n}(z)=B_{n,1}(z)$.
\end{cor}

\begin{rem}
\label{rem:Lyapunov spectrum symmetric about zero}While \cref{eq:group property of transfer matrices} alone does not allow us to conclude a symmetry property for the exponents, if we restrict to $z\in\mathbb{R}$, 
then the Hermitian symplectic condition implies that the exponents are symmetric about zero, that is, $\gamma_{j}(z)=-\gamma_{2N-(j-1)}(z)$
for all $j\in\{1,\dots,N\}$.
\begin{proof}
Since we are taking the logarithm, the symmetry of the singular values of the Hermitian symplectic matrix $B_{n}(z)$ about one (as shown in \cref{prop:Eigenvalues of Hermitian symplectic matrix are symmetric about S^1})
implies a symmetry of the exponents about zero.
\end{proof}
\end{rem}

\subsection{Irreducibility associated with the transfer matrices}\label{sec:irreducibility}
In order to prove the irreducibility of the semigroup generated by $\supp(\mu_\lambda)$, we study the map that takes a hopping matrix to the transfer matrix and its relation to $\HSG$.
\begin{prop}
	Let \begin{align}
	M = \begin{pmatrix}
	\lambda T^\circ & -T \\
	T^\circ & 0
	\end{pmatrix}\label{eq:general form of transfer matrix in irreducibility proof}
	\end{align}
	with $\lambda\in\RR\setminus\{0\}$. Then the map $$ \GL_N(\CC)^3\ni(T_1,T_2,T_3)\mapsto M_1 M_2 M_3 \in \HSG $$ is a submersion, i.e. its tangent map has maximal rank.
	\label{prop:product of three transfer matrices is a submersion}
\end{prop}

\begin{lem}
	\label{lem:General form of Hermitian symplectic matrix when D is invertible}
	Let $$ M=\begin{pmatrix}A & B\\
	C & D
	\end{pmatrix}\in \GL_{2N}(\mathbb{C}) $$ with $D\in \GL_{N}(\mathbb{C})$. Then $M\in\HSG$ iff 
	\begin{align}
	B=RD\,,\qquad C = DS\label{eq:DRS chart for open subset of Hermitian symplectic group}
	\end{align} 
	with $R=R^\ast$, $S=S^\ast$ and \begin{align}
	A = D^\circ (\mathds{1}+B^\ast C)\,.\label{eq:form for A block in a Hermitian symplectic matrix with D invertible}
	\end{align}
\end{lem}
The proof may be found in \cref{sec:The Hermitian symplectic group}.

Let $\mathscr{S}_0\subseteq\HSG$ be the open subset (and hence submanifold) given by $\mathscr{S}_0=\{M\in\HSG|\det D \neq 0\}$. By \cref{lem:General form of Hermitian symplectic matrix when D is invertible}, the map $\mathscr{S}_0\to \GL_N(\mathbb{C})\times \Herm_N(\mathbb{C})\times \Herm_N(\mathbb{C})$, $M\mapsto (D,R,S)$ is a coordinate chart of $\mathscr{S}_0$; for short $M\cong(D,R,S)$. In particular $$ \dim_\mathbb{R} \HSG = \dim_{\mathbb{R}} \mathscr{S}_0 = 2N^2+N^2+N^2=4N^2\,.$$
In the following $\lambda \in \mathbb{R}$, $\lambda\neq0$ is fixed.

Let $\mathscr{S}_1\subseteq\mathscr{S}_0$ be the submanifold which in terms of the chart consists of matrices $M\cong(D,R=\lambda\mathds{1},S)=:(D,S)$. In particular $\dim_{\mathbb{R}}\mathscr{S}_1 = 3N^2$. Moreover, we consider matrices of the form \cref{eq:general form of transfer matrix in irreducibility proof} with $T\in \GL_N(\mathbb{C})$. Then $M\in\HSG$, as remarked before \cref{eq:definition of the Hermitian symplectic group}, though they are not of the form discussed in \cref{lem:General form of Hermitian symplectic matrix when D is invertible} because of $D=0$. However, their products are, $M_1 M_2\in\mathscr{S}_1$, with 
\begin{align}
M_1 M_2 \cong (D,S)=(-T_1^\circ T_2,-\lambda(T_2^\ast T_2)^{-1})
\label{eq:DRS chart for product of two transfer matrices}
\end{align}
in terms of the chart. In fact
\begin{align*}
M_1 M_2 = \begin{pmatrix} \star & -\lambda T_1^\circ T_2 \\
\lambda T_1^\circ T_2^\circ & - T_1^\circ T_2\,,
\end{pmatrix}
\end{align*}
from which the claims about $D$ and $R$, i.e. $D\in\GL_N(\CC)$ and $R=\lambda\Id$, are evident by comparison with \cref{eq:DRS chart for open subset of Hermitian symplectic group}; the one about $S$ follows from $T_2 (T_2^\ast T_2)^{-1}=T_2^\circ$.

Let $M$ be as in \cref{eq:general form of transfer matrix in irreducibility proof} and $M'\cong(D,S)\in\mathscr{S}_1$. Then $MM'\in\mathscr{S}_0$ with $MM'\cong(\tilde{D},\tilde{R},\tilde{S})$, where 
\begin{align}
\tilde{D} = \lambda T^\circ D\,,\qquad \tilde{R} = \lambda - \lambda^{-1}|T^\ast|^2\,,\qquad\tilde{S} = S + \lambda^{-1} |D|^{-2}\,.\label{eq:DRS for product of three transfer matrices}
\end{align}
In fact, 
\begin{align*}
\begin{pmatrix}
\lambda T^\circ & -T \\
T^\circ & 0
\end{pmatrix} \begin{pmatrix}
A & \lambda D \\
D S & D
\end{pmatrix} = \begin{pmatrix}
\star & (\lambda^2 T^\circ - T) D \\
T^\circ A  & \lambda T^\circ D
\end{pmatrix}\,,
\end{align*}
where $A=D^\circ + \lambda D S$ by \cref{eq:form for A block in a Hermitian symplectic matrix with D invertible}. Then the claim \cref{eq:DRS for product of three transfer matrices} is obvious for $\tilde{D}$ and by \cref{eq:DRS chart for open subset of Hermitian symplectic group} amounts for the rest to $$ T^\circ A = \tilde{D}\tilde{S}\,,\qquad (\lambda^2 T^\circ - T)D = \tilde{R}\tilde{D}\,. $$
The two conditions simplify to $ D^{-1} D^\circ + \lambda S = \lambda \tilde{S}$ and to $\lambda^2 - T(T^\circ)^{-1} = \lambda \tilde{R}$, which by $D^{-1}D^\circ = |D|^{-2}$ are satisfied.

The next lemma concludes the proof of \cref{prop:product of three transfer matrices is a submersion}.

\begin{lem}
	The above product maps \begin{align*} \GL_N(\mathbb{C})\times \GL_N(\mathbb{C}) \to \mathscr{S}_1\,,&\qquad (T_1,T_2)\mapsto M_1 M_2=M'\,,\\\GL_N(\mathbb{C})\times\mathscr{S}_1 \to \mathscr{S}_0\,,&\qquad(T,M')\mapsto MM' \end{align*} are submersions.
	\begin{proof}
		We represent domain and codomain of both maps in their charts: 
		\begin{align}
		(T_1,T_2)\mapsto(D,S)\,,\qquad(T,D,S)\mapsto(\tilde{D},\tilde{R},\tilde{S})\label{eq:Product of transfer matrices two or three in the DRS chart}
		\end{align}
		with right hand sides given by \cref{eq:DRS chart for product of two transfer matrices,eq:DRS for product of three transfer matrices}. As a preparation we observe that the following maps are submersions: 
		\begin{itemize}
			\item $\GL_N(\mathbb{C})\toitself$, $T\mapsto T^{-1},T^\ast, TM$ or $MT$ for some fixed $M\in \GL_N(\mathbb{C})$,
			\item $\GL_N(\mathbb{C})\cap \Herm_{N}(\mathbb{C})\toitself$, $S\mapsto S^{-1}$
			\item $\GL_N(\mathbb{C})\to \Herm_{N}(\mathbb{C})$, $T\mapsto |T|^2$ or $|T^\ast|^2$,
			\item $\Herm_{N}(\mathbb{C})\toitself: S\mapsto S+M$ for some fixed $M\in\Herm_N(\CC)$.
		\end{itemize}
		Only the map $T\mapsto |T|^2$ may deserve comment. It has a differentiable right inverse $S\mapsto S^{1/2}$ on the open subset $\{S>0\}\subseteq \Herm_N(\mathbb{C})$, whence the claim.
		
%
		The first map \cref{eq:Product of transfer matrices two or three in the DRS chart} is then written as a concatenation $$ (T_1,T_2)\mapsto(D,T_2)\,,\qquad T_2\mapsto S$$ of maps that are seen to be submersions. Likewise for the second map: $$ (T,D,S)\mapsto (T,D,\tilde{S})\,,\qquad(T,D)\mapsto(T,\tilde{D})\,,\qquad T\mapsto \tilde{R}\,.$$
	\end{proof}
\end{lem}

\begin{prop}
	\label{prop:open subset of symplectic group}Let $\lambda\in\mathbb{R}\setminus\{0\}$. Then the semigroup $T_{\mu_{\lambda}}$ generated by $\supp \mu_\lambda$ contains an open subset of $\HSG$.
	\begin{proof}
		By definition, the measure $\mu_\lambda$ is the one induced by $\alpha_0, \alpha_1$ through \cref{eq:general form of transfer matrix in irreducibility proof} and the map $(M_1,M_2)\mapsto B:=M_2 M_1$. Now $T_{\mu_{\lambda}}$ contains for sure all matrices $B'B=M_2'M_1'M_2M_1$. By \cref{assu:probability measures support contains open subset}, \cref{prop:product of three transfer matrices is a submersion} and the submersion theorem \cite[Theorem 7.1]{Bredon_1993}, the matrices $M_1'M_2M_1$ cover an open subset of $\mathscr{S}_0$ and hence of $\HSG$, and so does $B'B$.
	\end{proof}
\end{prop}

\begin{rem}
For the usual Anderson model on a strip, as in \cite{KLEIN1990135,Simon1985}
for example, the transfer matrix is rather 
\begin{align*}
A_{x}(z) & =  \begin{pmatrix}z-V_{x} & -\mathds{1}_{N}\\
\mathds{1}_{N} & 0_{N}
\end{pmatrix}
\end{align*}
with $(V_{x})_{x\in\mathbb{Z}}$ the independent, identically
distributed sequence of onsite potentials (which in general should
take values in $\Herm_{N}(\mathbb{C})$). Then known
results, say, in \cite{KLEIN1990135} (and references therein) show
that the semigroup generated by the support of this transfer matrix
contains an open subset of the Hermitian symplectic group (for \emph{all}
$z\in\mathbb{C}$). Taking this result (which holds also for $z=0$) to replace our \cref{prop:open subset of symplectic group} and proving an a-priori bound, one could extend the analysis here to the Anderson model on the strip as well.
\end{rem}

\begin{cor}
\label{cor:The Lyapunov spectrum is simple}If $\lambda\in\mathbb{R}\backslash\{0\}$,
then the Lyapunov spectrum is simple: $\gamma_{j}(\lambda)\neq\gamma_{j'}(\lambda)$
for all $j\neq j'$ in $\{1,\dots,2N\}$. 
\begin{proof}
We apply \cite[Proposition IV.3.5]{Bougerol9781468491746}, which goes through even though it is applied on $Sp_{2N}(\mathbb{R})$ whereas here we apply it on $\HSG$.
\end{proof}
\end{cor}

\begin{cor}
\label{cor:Smallest Lyapunov exponent is strictly positive}
If $\lambda\in\mathbb{R}\backslash\{0\}$, then $\gamma_{N}(\lambda)>0$.
\begin{proof}
We always have $\gamma_{N}\geq\gamma_{N+1}$ because this is how we
choose the ordering of the labels. By \cref{rem:Lyapunov spectrum symmetric about zero}
we have that $\gamma_{N}(\lambda)=-\gamma_{N+1}(\lambda)$
and by \cref{cor:The Lyapunov spectrum is simple} $\gamma_{N}(\lambda)\neq\gamma_{N+1}(\lambda)$.
\end{proof}
\end{cor}

\begin{rem}
$\gamma_{N}(0)=0$ is possible though generically false. See also \cref{sec:chiral RBM}.
\begin{proof}
When $z=0$, via \cref{eq:Schroedinger equation formulated in terms of the transfer matrix},
for even sites:
\begin{align*}
\Psi_{2x} & \equiv  \begin{pmatrix}T_{2x+1}^{\ast}\psi_{2x+1}\\
\psi_{2x}
\end{pmatrix} =  A_{2x}(0)\Psi_{2x-1}\\
 & =  \begin{pmatrix}0 & -T_{2x}\\
T_{2x}^\circ & 0_{N}
\end{pmatrix}\begin{pmatrix}T_{2x}^{\ast}\psi_{2x}\\
\psi_{2x-1}
\end{pmatrix} =  \begin{pmatrix}-T_{2x}\psi_{2x-1}\\
\psi_{2x}
\end{pmatrix}\,.
\end{align*}
That is, 
\begin{align}
\psi_{2x+1} & =  -T_{2x+1}^\circ T_{2x}\psi_{2x-1}\,,\label{eq:zero energy transfer matrices}
\end{align}
and similarly for the odd sites, 
\begin{align*}
\psi_{2x+2} & =  -T_{2x+2}^\circ T_{2x+1}\psi_{2x}\,.
\end{align*}

As a result, within its zero eigenspace, $H$ commutes with the operator
$(-1)^{X}$ where $X$ is the position operator (this operator
gives the parity of the site) and the problem splits into two independent
first-order difference equations, with transfer matrices $-T_{2x+1}^\circ T_{2x}$
and $-T_{2x+2}^\circ T_{2x+1}$ respectively.
These transfer matrices are \emph{not} symplectic, nor is the absolute
value of their determinant equal to one\textendash they are merely
elements in $\GL_{N}(\mathbb{C})$, and only their direct
sum has these properties.

Hence the theorem insuring the simplicity of the Lyapunov spectrum,
\cite[pp. 78, Theorem IV.1.2]{Bougerol9781468491746} does not help
in this case, since the simplicity of the Lyapunov spectrum of these
transfer matrices will not imply that none of the exponents are zero.
Instead we are reduced to the more direct question of whether
any of the exponents are zero or not.

Here is a (trivial) example where there is a zero exponent: when $N=1$, there is only one exponent for each (separate) chirality sector, and the hopping matrices are merely complex numbers. Then for the (even) positive chirality e.g. we have
\begin{align*}
\gamma_{1}^+(0) & \equiv  \lim_{n\to\infty}\frac{1}{n}\mathbb{E}[\log(|(-T_{2n+1}^\circ T_{2n})(-T_{2n-1}^\circ T_{2n-2})\dots(-T_{3}^\circ T_{2})|)]\\
 & =  \lim_{n\to\infty}\frac{1}{n}(-\sum_{l=3,l\text{ odd}}^{2n+1}\mathbb{E}[\log(|T_{l}|)]+\sum_{l=2,l\text{ even}}^{2n}\mathbb{E}[\log(|T_{l}|)])\\
 &   (\text{independence property})\\
 & =  \mathbb{E}_{\alpha_{0}}[\log(|\cdot|)]-\mathbb{E}_{\alpha_{1}}[\log(|\cdot|)]\,.
\end{align*}
We immediately see that when $\alpha_{1}=\alpha_{0}$, the exponent is zero,
and when $\alpha_{1}\neq\alpha_{0}$ and both have non-zero $\log$-expectation
value (not the same value), then the exponent is non-zero.

In general, as can be seen from the formula for the transfer matrices,
if $\supp \alpha_{1} $ is concentrated within the unit
ball of $\GL_{N}(\mathbb{C})$ and $\supp \alpha_{0} $
is very far outside the unit ball of $\GL_{N}(\mathbb{C})$,
then we are guaranteed that none of the exponents will be zero at zero energy.
\end{proof}
\end{rem}

In the rest of this section we establish continuity properties of the exponents and finally connect $\gamma_N(\lambda)$ to the decay rate (in $n$) of a product of $n$ transfer matrices. This is a simple extension of the analysis in \cite[Section 2]{KLEIN1990135} to the complex-valued case with off-site randomness. Thus we frequently use the notation and conventions of \cite{KLEIN1990135} below sometimes without explicit reference.

We adopt the following viewpoint which is customary in the study of product of random matrices. Since we are analyzing matrices in $\HSG$, we'll be interested in \emph{isotropic} subspaces (the subspaces of $\mathbb{C}^{2N}$ on which $J$ restricts to the zero bilinear form) which are left invariant by $\HSG$. Correspondingly we study the \emph{isotropic} Grassmannian manifold, $\bar{L}_k$, the set of isotropic subspaces of dimension $k$ within $\mathbb{C}^{2N}$ with $k\in\{1,\dots,N\}$. A convenient way to parametrize such isotropic subspaces is via exterior powers: a simple vector $u_1\wedge\dots\wedge u_k\in\wedge^k\mathbb{C}^{2N}$ with $\{u_i\}_i$ linearly independent and $\langle u_i, J u_j\rangle = 0$ for all $i,j\in\{1,\dots,k\}$ defines a point in $\bar{L}_k$, and one can talk about the action of $\HSG$ on such a point by lifting $g\in\HSG$ to $\wedge^k g$ on $\wedge^k\mathbb{C}^{2N}$.

\begin{defn}
For any $[v]\in\bar{L}_k$
and $z\in\mathbb{C}$ let 
\begin{align*}
\Phi_{z}([v]) & :=  \mathbb{E}\bigl[\log\bigl(\frac{\norm{\wedge^k A_{1}(z)v}}{\|v\|}\bigr)\bigr]\,.
\end{align*}
\end{defn}

\begin{prop}\label{prop:cocyle continuous wrt vector}
$[v]\mapsto\Phi([v])$ is continuous.
\begin{proof}
This is \cite[Prop. 2.4]{KLEIN1990135} or \cite[V.4.7
(i)]{Carmona1990} for our model. Via \cref{prop:bound on log of norm of exterior product of matrix of |det|=00003D1}
and \cref{prop:log of transfer matrix is integrable} we find that
the ``sequence'' $[v]\mapsto\log(\frac{\norm{\wedge^k A_{1}v}}{\|v\|})$
is bounded by the integrable function $k(2N-1)\log(\norm{ A_{1}})$.
So using the Lebesgue dominated convergence theorem and the fact that
$[v]\mapsto\log(\frac{\norm{\wedge^k A_{1}v}}{\|v\|})$ is continuous
we find our result.
\end{proof}
\end{prop}

\begin{prop}\label{prop:cocycle lip cont wrt energy}
$z\mapsto\Phi_{z}([v])$ is Lipschitz continuous,
uniformly in $[v]$, as long as $z$ ranges in a compact subset of $\CC$.
\begin{proof}
This is \cite[Prop. 2.4]{KLEIN1990135}  or \cite[V.4.7
(ii)]{Carmona1990} for our model. First note that since $\norm{Mu}=\norm{ML^{-1}Lu}\leq\norm{ML^{-1}}\norm{Lu}$
we have for all $[v]$,
\begin{align*}
\Phi_{z}([v])-\Phi_{w}([v]) & \leq  \mathbb{E}[\log(\norm{A_{1}(z)(A_{1}(w))^{-1}})]\,.
\end{align*}
By symmetry, 
\begin{align*}
\Phi_{w}([v])-\Phi_{z}([v]) & \leq  \mathbb{E}[\log(\norm{A_{1}(w)(A_{1}(z))^{-1}})]\\
 &  \quad (|\det(A_{1}(w)(A_{1}(z))^{-1})|=1)\\
 & =  (2N-1)\mathbb{E}[\log(\norm{A_{1}(z)(A_{1}(w))^{-1}})]\,,
\end{align*}
so that 
\begin{align*}
|\Phi_{z}([v])-\Phi_{w}([v])| & \leq  (2N-1)\mathbb{E}[\log(\norm{A_{1}(z)(A_{1}(w))^{-1}})]\,.
\end{align*}
Next we have 
\begin{align*}
S(z)a(X)S(z)a(Y)[S(w)a(X)S(w)a(Y)]^{-1} & =  S(z)a(X)S(z)S(w)^{-1}a(X)^{-1}S(w)^{-1}\,.
\end{align*}
We remark that 
\begin{align*}
S(z)S(w)^{-1} & =  \begin{pmatrix}z & -\mathds{1}_{N}\\
\mathds{1}_{N} & 0_{N}
\end{pmatrix}\begin{pmatrix}0_{N} & \mathds{1}_{N}\\
-\mathds{1}_{N} & w
\end{pmatrix} =  \begin{pmatrix}\mathds{1}_{N} & (z-w)\mathds{1}_{N}\\
0 & \mathds{1}_{N}
\end{pmatrix}\\
 & =  \mathds{1}_{2N}+(z-w)\begin{pmatrix}0 & \mathds{1}_{N}\\
0 & 0
\end{pmatrix}\,,
\end{align*}
so that 
\begin{align*}
\norm{S(z)a(X)S(z)S(w)^{-1}a(X)^{-1}S(w)^{-1}} & =  \norm{S(z)a(X)[\mathds{1}_{2N}+(z-w)\begin{pmatrix}0 & \mathds{1}_{N}\\
0 & 0
\end{pmatrix}]a(X)^{-1}S(w)^{-1}}\\
 & \leq  1+|z-w|\norm{\begin{pmatrix}0 & \mathds{1}_{N}\\
0 & 0
\end{pmatrix}}+\\
 &   +|z-w|\norm{S(z)}\norm{a(X)}\times\\
 &   \times\norm{\begin{pmatrix}0 & \mathds{1}_{N}\\
0 & 0
\end{pmatrix}}\norm{a(X)^{-1}}\norm{S(w)^{-1}}\,.
\end{align*}
Now $\norm{\begin{pmatrix}0 & \mathds{1}_{N}\\
0 & 0
\end{pmatrix}}=1$, $\norm{a(X)^{-1}}=\norm{a(X)}$ and $\log(1+x)\leq x$
for all $x\in\mathbb{R}$ so that we find 
\begin{align*}
\log(\norm{A_{1}(z)(A_{1}^{w})^{-1}}) & \leq  |z-w|(1+\norm{S(z)}\norm{S(w)^{-1}}\norm{a(X)}^{2})\,.
\end{align*}
Since $z\mapsto\norm{S(z)}$ is continuous, $z$ ranges in a compact
set, and using \cref{assu:regularity of probability distributions}, we find 
\begin{align*}
|\Phi_{z}([v])-\Phi_{w}([v])| & \leq  |z-w|(1+\sup_{z'}\norm{S_{z'}}^{2}\mathbb{E}[\norm{a(T_{1})}^{2}])\,,
\end{align*}
obtaining the claim for $k=1$; the other cases being an easy generalization.
\end{proof}
\end{prop}

We state without proof the following corollary and proposition which may be found in \cite[Prop. 2.5]{KLEIN1990135}, \cite[Prop. 2.6]{KLEIN1990135} respectively. Their proof in the present setting relies on \cref{prop:cocyle continuous wrt vector,prop:cocycle lip cont wrt energy}.

\begin{cor}
The map $z\mapsto\gamma_{j}(z)$ is continuous for all
$j$ as $z$ ranges in a compact subspace of $\CC$.
\end{cor}

%

\begin{prop}
\label{prop:Convergence of Lyapunov limit uniformly}For each $j\in\{1,\dots,N\}$
and $[x]\in\bar{L}_j$,
we have 
\begin{align*}
\lim_{n\to\infty}\frac{1}{n}\mathbb{E}[\log(\frac{\norm{\wedge^{j}B_{n}(\lambda)x}}{\|x\|})] & =  \sum_{l=1}^{j}\gamma_{l}(\lambda)\,,
\end{align*}
the limit being \emph{uniform} as $\lambda$ ranges in a compact subspace
of $\mathbb{R}$ and also \emph{uniform} in $[x]$.
\end{prop}
The following result finally connects $\gamma_N(\lambda)>0$ with the exponential decay of a product of transfer matrices. Its proof is included in the appendix for the reader's convenience, as it appeared merely as an outline in \cite[Prop. 2.7]{KLEIN1990135}.
\begin{prop}
\label{prop:fractional moments of products of transfer matrices decay exponentially}
For any compact $K\subseteq\mathbb{R}$, $j\in\{1,\dots,N\}$
such that $\gamma_{j}(\lambda)>0$ for all $\lambda\in K$, there exist
some $s\in(0,1)$, $N\in\mathbb{N}$ and $C>0$
such that 
\begin{align*}
\mathbb{E}[(\frac{\norm{\wedge^{j-1}B_{n}(\lambda)y}}{\|y\|}\frac{\|x\|}{\norm{\wedge^{j}B_{n}(\lambda)x}})^{s}] & \leq  \exp(-C n)
\end{align*}
for all $n\in\mathbb{N}_{\geq N}$, for all $[x]\in \bar{L}_{j}$ and for all $[y]\in \bar{L}_{j-1}$.
\end{prop}

\section{An a-priori bound}
\label{sec:a-priori bound}
In this short section we establish the a-priori boundedness of one-step Green's functions, which is a staple of the fractional-moments method. The fact that for our model one uses the one-step Green's function rather than the diagonal one is due to the fact we do not have on-site randomness, which forces the usage of rank-2 perturbation theory.
\begin{prop}
\label{prop:A-priori bound}For any $s\in(0,1)$ we
have some strictly positive constant $C_{s}$ such that 
\begin{align*}
\mathbb{E}[\norm{G(x,x-1;\,z)}^{s}] & \leq  C_{s}<\infty\,,
\end{align*}
for all $x\in\mathbb{Z}$ and for all $z\in\mathbb{C}$. So $C_{s}$
depends on $s$, $\alpha_{1}$ and $\alpha_{0}$.
\begin{proof}
Using finite-rank perturbation theory, we can find the explicit dependence
of the complex number $G(x,x-1;\,z)_{i,j}$ (for some
$(i,j)\in\{1,\dots,N\}^{2}$) on the random hopping
$(T_{x})_{i,j}$. Indeed, $(T_{x})_{i,j}=\left\langle \delta_{x}\otimes e_{i},H(T),\delta_{x-1}\otimes e_{j}\right\rangle $.
We find that $G(x,x-1;\,z)_{i,j}$ is equal to the
bottom-left matrix element of the $2\times2$ matrix 
\begin{equation}
\bigl(A+\begin{pmatrix}0 & \overline{\lambda}\\
\lambda & 0
\end{pmatrix}\bigr)^{-1}\,,\label{eq:matrix for a-priori bound}
\end{equation}
where $A$ is the inverse of the $2\times2$ matrix 
\[
\begin{pmatrix}\tilde{G}(x-1,x-1;\,z)_{j,j} & \tilde{G}(x-1,x;\,z)_{j,i}\\
\tilde{G}(x,x-1;\,z)_{i,j} & \tilde{G}(x,x;\,z)_{i,i}
\end{pmatrix}\,,
\]
with $\tilde{H}$ being $H$ with $(T_{x})_{i,j}$ ``turned
off'' (i.e. set to zero), and $\lambda:=(T_{x})_{i,j}$ for convenience.
Thus our goal is to bound the fractional moments with respect to $\lambda\in\mathbb{C}$
of the off-diagonal entry of the matrix in \cref{eq:matrix for a-priori bound}
where $A$ is some given $2\times2$ matrix with complex entries.
The only thing we know about $A$ is that $\Im\{A\} >0$
though we won't actually use this (cf. \cite[Lemma 5]{Graf1994}). 

First note that $\begin{pmatrix}0 & \overline{\lambda}\\
\lambda & 0
\end{pmatrix}=\Re\left\{ \lambda\right\} \sigma_{1}+\Im\left\{ \lambda\right\} \sigma_{2}$ and also that for any $2\times2$ matrix $M$ we have $\frac{1}{2}tr(M\sigma_{1})=\frac{1}{2}(M_{12}+M_{21})$
whereas $\frac{1}{2}tr(M\sigma_{2})=\frac{i}{2}(M_{12}-M_{21})$.
Thus by the triangle inequality,
\begin{align*}
\mathbb{E}[|M_{12}|^{s}] & =  \mathbb{E}[|\frac{1}{2}tr(M\sigma_{1})|^{s}]+\mathbb{E}[|\frac{1}{2}tr(M\sigma_{2})|^{s}]\,.
\end{align*}
So, if we get control on each summand separately we could bound
$\mathbb{E}[|M_{12}|^{s}]$.

We write 
\begin{align*}
A+\begin{pmatrix}0 & \overline{\lambda}\\
\lambda & 0
\end{pmatrix} & =  \tilde{A}+\Re\left\{ \lambda\right\} \sigma_{1}
\end{align*}
and we expand $\tilde{A}$ as $\tilde{A}=a_{0}+a_{i}\sigma_{i}$ for
some $\begin{pmatrix}a_{0}\\
a
\end{pmatrix}\in\mathbb{C}^{4}$, so that if $M:=(a_{0}+a_{i}\sigma_{i}+\Re\left\{ \lambda\right\} \sigma_{1})^{-1}$
we find
\begin{align*}
\frac{1}{2}\tr(M\sigma_{1}) & =  \frac{-a_{1}-\Re\left\{ \lambda\right\} }{a_{0}\,^{2}-a_{2}\,^{2}-a_{3}\,^{2}-(a_{1}+\Re\left\{ \lambda\right\} )^{2}}\,.
\end{align*}
We now apply \cref{prop:Improved Lemma for a-prior bound} just below with $z=-a_{1}-\Re\left\{ \lambda\right\} $,
$c=a_{0}\,^{2}-a_{2}\,^{2}-a_{3}\,^{2}$ to get the result (using
the Layer-Cake representation), which relies on \cref{assu:probability measures are uniformly tau Hoelder continuous}. The term $\mathbb{E}[|\frac{1}{2}tr(M\sigma_{2})|^{s}]$
is dealt with in precisely the same way.
\end{proof}
\end{prop}

\begin{prop}
	\label{prop:Improved Lemma for a-prior bound}Consider the subset
	of the complex plane 
	\begin{align*}
		D_{c,t} & :=  \{z\in\mathbb{C}||\frac{z}{z^{2}-c^{2}}|>t\}
	\end{align*}
	for $c\in\mathbb{C}$ and $t>0$ and the line $\mathbb{R}\ni\lambda\mapsto z:=\alpha\lambda+\beta$
	($\alpha,\beta\in\mathbb{C}$, $|\alpha|=1$) parametrized
	by arclength. Then its intersection with $D_{c,t}$ is bounded in
	Lebesgue measure as 
	\begin{align*}
		|\{\lambda\in\mathbb{R}|z\in D_{c,t}\}| & \leq  \frac{4}{t}\,.
	\end{align*}
	\begin{proof}
		Since the statement is invariant w.r.t. rotations of $c,\alpha,\beta$
		about the origin, we may assume $c\geq0$. We then estimate the measure
		when $D_{c,t}$ is replaced by its intersection with the right half-plane
		$\{z\in\mathbb{C}|\Re\{z\} \geq0\}$ (and likewise for
		the left one). Then 
		\[
		|\frac{z}{z^{2}-c^{2}}|=\frac{1}{|z-c|}|\frac{z}{z+c}|\leq\frac{1}{|z-c|}
		\]
		because $\Re\{z\} +c\geq\Re\{z\} $ there,
		which implies $D_{c,t}\cap\{z|\Re\{z\} \geq0\}\subseteq\{z||z-c|<t^{-1}\}$.
		The intersection of that disk with any line is of length $2t^{-1}$
		at most.
	\end{proof}
\end{prop}

\begin{cor}
\label{cor:a-priori bound for xx green's function}For any $s\in(0,1)$
we have some strictly positive constant $C_{s}$ (not the same as
the one from above) such that 
\begin{align*}
\mathbb{E}[\norm{G(x,x;\,z)}^{s}] & <  \frac{1}{|z|}C_{s}\,,
\end{align*}
for all $x\in\mathbb{Z}$ and for all $z\in\mathbb{C}\backslash\{0\}$.
$C_{s}$ depends only on $s\in(0,1)$, $\alpha_{1}$
and $\alpha_{0}$.
\begin{proof}
From the relation $(H-z)R(z)\equiv\mathds{1}$
we find 
\begin{align*}
-zG(x,x)+T_{x+1}^{\ast}G(x+1,x;\,z)+T_{x}G(x-1,x;\,z) & =  \mathds{1}_{N}\,,
\end{align*}
so that 
\begin{align*}
G(x,x;\,z) & =  z^{-1}[\mathds{1}_{N}-T_{x+1}^{\ast}G(x+1,x;\,z)+T_{x}G(x-1,x;\,z)]\,.
\end{align*}
Hence by the triangle inequality, H\"older's inequality, \cref{assu:regularity of probability distributions}
and \cref{prop:A-priori bound} we find the result.
\end{proof}
\end{cor}

\begin{rem}
The last two statements hold equally well if we replace $G(x,x)$
or $G(x-1,x)$ with the half-line Green's function or even a finite-volume Green's function.
\end{rem}

\section{Localization at non-zero energies}
\label{sec:loc at non-zero energies}
In this section we establish localization for all non-zero energies. We do this in two steps: first at real energies (and hence finite volume) due to the fact that the Furstenberg analysis requires the transfer matrices to be $\HSG$-valued which needs the energy to be real. We then extend this exponential decay off the real axis using the harmonic properties of the Green's function to get polynomial decay at complex energies, which in turn implies exponential decay via a decoupling-type lemma. Once localization of finite volume at complex energies is established, the infinite volume result is implied via the strong-resolvent convergence of the finite volume Hamiltonian to the infinite volume one.
\subsection{Localization at finite volume and real energies}

\begin{thm}
\label{thm:finite volume real energy FMC}
For any compact $K\subset\RR\setminus\{0\}$,
there is some $s\in(0,1)$ and $\mu>0$ such
that
\begin{align*}
\sup_{\lambda\in K}\mathbb{E}[|G_{[x,y]}(x,y;\,\lambda)|^{s}] & <  |\lambda|^{-s}e^{-\mu|x-y|}\,,
\end{align*}
for all $(x,y)\in\mathbb{Z}^{2}$ with $|x-y|$
sufficiently large. Here $G_{[x,y]}$ is the Green's function of the finite-volume restriction of $H$ to $[x,y]\subseteq\mathbb{Z}$.
\begin{proof}
Let $\lambda\in K$ be given and $s\in(0,1)$. By \cref{lem:A Green's function formula}
we know that 
\begin{align*}
\norm{G_{[x,y]}(x,y;\,\lambda)}^{s} & \leq  |\lambda|^{-s}C(z)^{\frac{s}{2}}(\sum_{j=1}^{N}\frac{\norm{\wedge^{N-1}B_{y-1,x-1}(\lambda)u_{j}}^{2}}{\norm{\wedge^{N}B_{y-1,x-1}(\lambda)u}^{2}})^{\frac{s}{2}}\\
 & \leq  |\lambda|^{-s}C(\lambda)^{\frac{s}{2}}\sum_{j=1}^{N}\frac{\norm{\wedge^{N-1}B_{y-1,x-1}(\lambda)u_{j}}^{s}}{\norm{\wedge^{N}B_{y-1,x-1}(\lambda)u}^{s}}\,.
\end{align*}
Here $u=e_1\wedge\dots\wedge e_N$ and $u_j=e_1\wedge\dots\wedge e_{j-1}\wedge e_{j+1}\dots\wedge e_N$. Now using H\"older's inequality we get 
\begin{align*}
\mathbb{E}[\norm{G_{[x,y]}(x,y;\,\lambda)}^{s}] & \leq  |\lambda|^{-s}\mathbb{E}[C(\lambda)^{s}]^{\frac{1}{2}}\sum_{j=1}^{N}\mathbb{E}[\frac{\norm{\wedge^{N-1}B_{y-1,x-1}(\lambda)u_{j}}^{s}}{\norm{\wedge^{N}B_{y-1,x-1}(\lambda)u}^{s}}]^{\frac{1}{2}}\,.
\end{align*}
$\mathbb{E}[C(\lambda)^{s}]^{\frac{1}{2}}$ is bounded
uniformly in $\lambda$ using \cref{prop:A-priori bound}, \cref{cor:a-priori bound for xx green's function}
and \cref{assu:regularity of probability distributions}. Now $u_{j}\in \bar{L}_{N-1}$
for any $j$ and $u\in \bar{L}_{N}$; $K\notin0$, $\gamma_{N}(\lambda)>0$
for all $\lambda\in K$
so that we may apply \cref{prop:fractional moments of products of transfer matrices decay exponentially}
to get that for $s_{K}$, for appropriate $\infty>C>0$, for $|x-y|>n_{K}$
\begin{align*}
\mathbb{E}[\norm{G_{[x,y]}(x,y;\,\lambda)}^{s_{K}}] & <  |\lambda|^{-s_{K}}CN\exp(-\frac{1}{2}C_{j,K}|x-y|)\,,
\end{align*}
which implies the bound in the claim. Note that we have used stationarity
of $\mathbb{P}$ to go from $0,x-y$ to $y,x$. 
\end{proof}
\end{thm}
\subsection{Infinite volume complex energy}
Obtaining polynomial decay off the real axis from exponential decay at the real axis was already accomplished in \cite[Theorem 4.2]{Aizenman2001} using properties of the Poisson kernel. Here we provide another proof of this fact and go on to show that \emph{any} decay implies exponential decay for our one dimensional models.

Let $Q(a)\subseteq\mathbb{C}$ be an open square of side $a>0$; its lower side is placed on the real axis $\mathbb{R}\subseteq\mathbb{C}$ with endpoints denoted $x_\pm$, ($x_+-x_-=a$). By $\tilde{Q}(a)\subseteq Q(a)$ we mean the symmetrically placed subsquare $\tilde{Q}(a)=Q(a/2)$, see \cref{fig:square on the real axis for subharmonicity}.

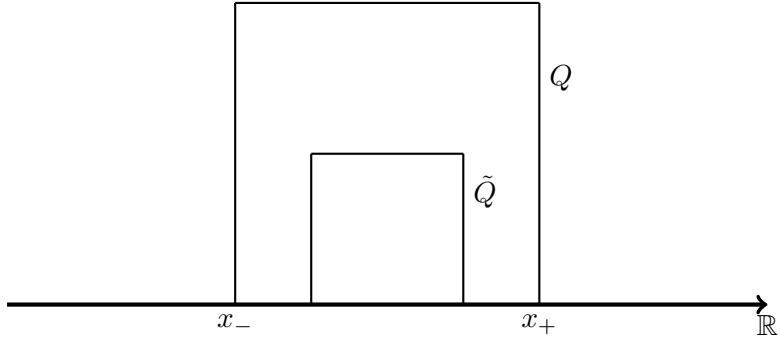
\begin{figure}[h]
	\centerline{\begin{tikzpicture}  
		\draw [thick] (-2,0) -- (-2,4);
		\draw [thick] (2,0) -- (2,4);
		\draw [thick] (-2,4) -- (2,4);
		\draw [thick] (-1,0) -- (-1,2);
		\draw [thick] (1,0) -- (1,2);
		\draw [thick] (-1,2) -- (1,2);		
		\draw [ultra thick,->] (-5,0) -- (5,0);
		\node [below] at (5,0) {$\mathbb{R}$};
		\node [below] at (-2,0) {$x_-$};
		\node [below] at (2,0) {$x_+$};
		\node [right] at (2,3) {$Q$};
		\node [right] at (1,1.5) {$\tilde{Q}$};
		\end{tikzpicture}}
	
	\caption{\label{fig:square on the real axis for subharmonicity}The square $Q(a)$ in $\mathbb{C}$.}
\end{figure}

In this section $f=f(z)$ is a subharmonic function defined on $Q$. By its boundary values we simply mean $$f(z) = \limsup_{z'\to z}f(z')\quad(\leq+\infty,z\in\partial Q)\,.$$
We will use the maximum principle in the form 
\begin{align}
	f(z) \leq \sup_{w\in\partial Q} f(w)\,,\quad(z\in Q)\,.\label{eq:subharmonicity the maximum principle}
\end{align}
where, as usual, $0<s<1$.

We present two lemmas. The first one says that if $f$ is bounded everywhere on $\partial Q$, except for some controlled divergence when the real axis is approached, then $f$ is bounded on $\tilde{Q}$ with explicit bounds, i.e. not just by compactness. The second lemma says that if $f$ is small everywhere on $\partial Q$, except very near the real axis, where it is just bounded, then $f$ is small on $\tilde{Q}$. The two lemmas may be used in concatenation.
\begin{lem}\label{lem:subharmonicity f bounded on Q with controlled divergence then f is bounded on Q tilde}
	Let $M\geq0$ and suppose
	\begin{equation}
	\begin{aligned}
		f(z) &\leq M\,,\quad (z\in(\partial Q)_h)\,,\\
		f(z) &\leq M\bigl(\frac{a}{\Im\{z\}}\bigr)^s\,,\quad(z\in(\partial Q)_v)\,,\label{eq:subharmonicity bounds on f on the boundary of Q}
	\end{aligned}
	\end{equation}
	where $(\partial Q)_{h/v}$ are the horizontal / vertical parts of $\partial Q$. Then
	\begin{align*}
		f(z) \leq M(1+s^{1-s})\,,\quad(z\in\tilde{Q})\,.
	\end{align*}
\end{lem}
\begin{lem}\label{lem:subharmonicity f is small everywhere on the boundary except near the real axis then it is small in tilde Q}
	Let $m,M\geq0$ and let $I\subseteq(\partial Q)_v$ be the union of the two vertical intervals of length $b\leq a$ next to $x_\pm$. 
	Suppose \begin{equation}\begin{aligned}
		f(z) & \leq m \,, \quad (z\in\partial Q \setminus I)\,,\\
		f(z) & \leq M \,, \quad (z\in I)\,.\label{eq:subharmonicity bounds on f on the boundary of Q tilde}
	\end{aligned}\end{equation}
	Then \begin{align*}
		f(z) \leq 2^{1-s} M \bigl(\frac{b}{a}\bigr)^s + m\,.
	\end{align*}
\end{lem}
As a preliminary to the proofs, we consider the function $$ v(z) := -\Im\{z^{-s}\}$$ defined on the first quadrant $\{z\in\mathbb{C}|\Re\{z\}>0,\Im\{z\}>0\}$. It is harmonic and the chosen branch is made clear using polar coordinates $z=r e^{i\theta}$, ($r>0,0<\theta<\pi/2$) as $$ v(z) = r^{-s} \sin (s \theta)\,. $$
In particular, $0<v(z)<r^{-s}\sin(\frac{\pi}{2}s)$. We also note its boundary values $$ v(x)=0\,,\quad v(iy)=\sin(\frac{\pi}{2} s) y^{-s}\,,\quad(x,y>0)\,. $$
The analogous function on the second quadrant is $v(-\bar{z})$. 

We will also make use of the harmonic function $$h(z) := \frac{C}{\sin(\frac{\pi}{2}s)}\bigl(v(z-x_-)+v(x_+-\bar{z})\bigr)\,,\quad(z\in Q)\,,$$ for some $C>0$ and boundary values $$h(z) \geq C (\Im z)^{-s}\,,\quad(z\in(\partial Q)_v)\,.$$
Moreover we have $$|z-x_\pm|\geq\frac{a}{2}\,,\quad(z\in\tilde{Q})\,,$$ which yields the upper bound 
\begin{align}
h(z) \leq 2^{1-s} C a^{-s}\,,\quad(z\in\tilde{Q})\,.\label{eq:subharmonicity bound on h}
\end{align}
\begin{proof}[Proof of \cref{lem:subharmonicity f bounded on Q with controlled divergence then f is bounded on Q tilde}]
We have \begin{align*}
h(z) & \geq M \bigl(\frac{a}{\Im z})^s\,,\quad(z\in(\partial Q)_v)
\end{align*}
for $C:=M a^s$. Since $h(z),M\geq0$ anyway we have by \cref{eq:subharmonicity bounds on f on the boundary of Q} $$f(z) \leq h(z) + M\,,\quad(z\in\partial Q)\,.$$
Since the difference of the two sides is still subharmonic, the inequality applies to $z\in Q$ by the maximum principle \cref{eq:subharmonicity the maximum principle}. In particular, for $z\in\tilde{Q}$ we have $f(z) \leq 2^{1-s}M + M $ by \cref{eq:subharmonicity bound on h} as claimed.
\end{proof}
\begin{proof}[Proof of \cref{lem:subharmonicity f is small everywhere on the boundary except near the real axis then it is small in tilde Q}]
	We have $$h(z)\geq M\,,\quad(z\in I)$$ for $C b^{-s} := M$. So $$f(z)\leq h(z)+m\,,\quad(z\in\partial Q)$$ by \cref{eq:subharmonicity bounds on f on the boundary of Q tilde} and, as before, $$ f(z) \leq 2^{1-s} M \bigl(\frac{b}{a}\bigr)^s+m\,,\quad(z\in\tilde{Q})\,.$$
\end{proof}

We now turn to a slight generalization of the well-known Combes-Thomas estimate \cite{Combes_Thomas_1973}. The generalization is that we do \emph{not} assume that $\alpha_i$ have compact support such that $\sup_{x} \|T_{x}\| \leq K\,\exists K>0$.
\begin{prop}\label{prop:the Combes-Thomas estimate}
	We have  some $C>0$ such that
	\begin{align}
		\EE[\| G(x,y;E+i\eta)\|^s]\leq\frac{2}{\eta} e^{-C \eta |x-y|}
	\end{align} 
	for all $x,y\in\mathbb{Z}$, for all $\eta>0$ and $E\in\mathbb{R}$, regardless of the fate of localization at $E$.
	\begin{proof}
		The first step is to note that because of finite-rank perturbation
		theory, we have 
		\begin{align*}
			G(x,y;z) & =  (\Id-G(x,x-1;z)T_{x}^{\ast})G^{[x,y]}(x,y;z)(\Id-T_{y+1}^{\ast}G^{[x,\infty)}(y+1,y;z))
		\end{align*}
		so that with the a-priori bound \cref{prop:A-priori bound} on $\mathbb{E}[\norm{G(x,x-1;z)}^{s}]$,
		we only need to consider the finite-volume Green's function $G^{[x,y]}(x,y;z)$.
		Throughout we use $z=E+\ii\eta $.
		
		The second step is to divide the expectation value, for arbitrary
		$M>0$, using the trivial bound $\norm{G^{[x,y]}(x,y;z)}^{s}<\eta^{-s}$,
		and the original Combes-Thomas bound (see below) on \emph{bounded} hopping terms:
		\begin{align*}
			\mathbb{E}[\norm{G^{[x,y]}(x,y;z)}^{s}] & =  \int_{\{\omega\in\Omega|\sup_{x'\in[x,y]}\norm{T_{x'}(\omega)}<M\}}\norm{G^{[x,y]}(x,y;z)}^{s}\dif{\mathbb{P}(\omega)}+\\&\quad+\eta^{-s}(1-\mathbb{P}(\{\omega\in\Omega|\sup_{x'\in[x,y]}\norm{T_{x'}(\omega)}<M\}))\\
			& \leq  \frac{2^{s}}{\eta^{s}}\exp(-s\frac{\eta}{8M}\left|x-y\right|)+\eta^{-s}(1-\mathbb{P}(\{\omega\in\Omega|\sup_{x'\in[x,y]}\norm{T_{x'}(\omega)}<M\}))\,.
		\end{align*}
		Now we have using the i.i.d. property (we ignore the fact $T_{1}$ and $T_{2}$
		are distributed differently since it only adds notation clutter, but does not change the argument),
		\begin{align*}
			\mathbb{P}(\{\omega\in\Omega|\sup_{x'\in[x,y]}\norm{T_{x'}(\omega)\}<M}) & =  \mathbb{P}(\{\omega\in\Omega|\norm{T_{1}(\omega)}<M\})^{\left|x-y\right|}\\
			& =  (1-\mathbb{P}(\{\omega\in\Omega|\norm{T_{1}(\omega)}\geq M\}))^{\left|x-y\right|}
		\end{align*}
		Now the usual Markov inequality is 
		\begin{align*}
			\mathbb{P}(\{\omega\in\Omega|\norm{T_{1}(\omega)}\geq M\}) & \leq  \frac{\mathbb{E}[\norm{T_{1}(\omega)}]}{M}\,.
		\end{align*}
		However, it is actually true that for any $\alpha\geq0$, 
		\begin{align*}
			\mathbb{P}(\{\omega\in\Omega|\norm{T_{1}(\omega)}\geq M\}) & \leq  \frac{\mathbb{E}[\norm{T_{1}(\omega)}^{\alpha}]}{M^{\alpha}}\,.
		\end{align*}
		Hence we get, since $q\mapsto1-(1-q)^{n}$ is increasing
		for all $n$ and $q\ll1$ ($M$ big), 
		\begin{align*}
			1-\mathbb{P}(\{\omega\in\Omega|\sup_{x'\in[x,y]}\norm{T_{x'}(\omega)}<M\}) & \leq  1-(1-\frac{\mathbb{E}[\norm{T_{1}(\omega)}^{\alpha}]}{M^{\alpha}})^{\left|x-y\right|}\\
			& \leq |x-y| \frac{\mathbb{E}[\norm{T_{1}(\omega)}^{\alpha}]}{M^{\alpha}} \,.
		\end{align*}
		If we now pick $\alpha:=2$ and $M:=\mathbb{E}[\norm{T_{1}(\omega)}^{2}]^{-\frac{1}{2}}\left|x-y\right|^{\frac{3}{4}}$
		we find that 
		\begin{align*}
			\mathbb{E}[\norm{G^{[x,y]}(x,y;z)}^{s}] & \stackrel{\left|x-y\right|\to\infty}{\longrightarrow}  0\,.
		\end{align*}
		Now since we know that \emph{any} decay of $\mathbb{E}[\norm{G^{[x,y]}(x,y;z)}^{s}]$
		implies exponential decay for our model, \cref{lem:any decay implies exponential decay}, we find the desired result.

		For completeness we give now the proof of the usual Combes-Thomas estimate with bounded hopping (bounded by some $K<\infty$):

		Without loss of generality let $E=0$ and pick some $\mu>0$ (to be specified later). Also pick some $f$ bounded and Lipschitz such that $|f(x)-f(y)|\leq\mu|x-y|$. We use the notation $R(z)\equiv(H-z\mathds{1})^{-1}$ for the resolvent and also define $H_f := e^{f(X)} H e^{-f(X)}$ with $X$ the position operator and $B := H_f -H$. 
		
		Then we have $(H_f \psi)(x) = e^{f(x)-f(x+1)}T_{x+1}^\ast\psi(x+1)+e^{f(x)-f(x-1)}T_x\psi(x-1)$ so that the matrix elements of $B$ are given by $$ B_{x,y} = (e^{f(x)-f(x+1)}-1)T_{x+1}^\ast\delta_{y,x+1} + (e^{f(x)-f(x-1)}-1)T_x\delta_{y,x-1}\,, $$ whence it follows by Holmgren that $$ \| B \| \leq 2 (e^\mu-1)K =: b(\mu)\,, $$ 
		If we now choose $\mu := \log(1+\frac{\eta}{4K})$ then evidently $b(\mu) = \frac{1}{2}\eta$ so that $\|(H_f-z\mathds{1})\psi\|\geq(\eta-\|B\|)\|\psi\|$ and we get $\|(H_f-z\mathds{1})^{-1}\|\leq\frac{2}{\eta}$. But $\|G(x,y;E+i\eta)\|=|e^{-f(x)+f(y)}|\|\langle \delta_x,(H_f-z\mathds{1})^{-1}\delta_y\rangle\|\leq|e^{-f(x)+f(y)}|\frac{2}{\eta}$ and we obtain the result by the freedom in choice between $f$ and $-f$.
	\end{proof}
\end{prop}

We proceed to obtain the decay of the Green's function uniformly in $\eta\equiv\Im z$.
\begin{prop}\label{prop:decay of greens function above the real axis}
The finite-volume Green's function $\mathbb{E}[\|G_{[x,y]}(x,y;z)\|^s]$ decays in $|x-y|$ at any value of $\Im z$ for all $\Re z \neq 0$.
\begin{proof}
	Let $\lambda\in\mathbb{R}\setminus\{0\}$ be given and define $z:=\lambda+i\eta$. Define $f(z) := \mathbb{E}[\|G_{[x,y]}(x,y;z)\|^s]$ for fixed $x,y$ and $0<s<1$ sufficiently small such that the hypothesis for \cref{thm:finite volume real energy FMC} holds.
	
	Since the Green's function is holomorphic, it follows that $\|G_{[x,y]}(x,y;z)\|^s$ is subharmonic and hence so is $f$. 
	
	Let $0<a<4K$ be such that the interval about $\lambda$ does not include zero: $0\notin(\lambda-a,\lambda+a)$. Then due to \cref{thm:finite volume real energy FMC} and the basic fact that $\|G(x,y;z)\|\leq|\Im z|^{-1}$, we know  that there is some $M\geq0$ such that the assumptions of \cref{lem:subharmonicity f bounded on Q with controlled divergence then f is bounded on Q tilde} are fulfilled ($M$ depends on $a$ and the constants provided by \cref{thm:finite volume real energy FMC}). Hence we may conclude that $f(z) \leq M(1+2^{1-s})$ as $z$ ranges in $\tilde{Q}(a)$.
	
	Now pick any $b<a/2$. With $I := (\tilde{Q}(a))_v\cap\{z|\Im z \leq b\}$, \cref{prop:the Combes-Thomas estimate,thm:finite volume real energy FMC} imply that on $f(z)\leq m$ for all $z\in\partial \tilde{Q}(a)\setminus I$, with $m = C b^{-s} e^{-b d}$ with $d:=|x-y|$ large enough, for some constant $C>0$ (where we have used that $\log(1+\beta)\geq \alpha \beta$ for all $\beta\leq\alpha^{-1}$ and $\alpha < 1$). We also have still from \cref{lem:subharmonicity f bounded on Q with controlled divergence then f is bounded on Q tilde} that $f(z) \leq M(1+2^{1-s})$ for all $z\in I\subseteq \tilde{Q}(a)$. Hence \cref{lem:subharmonicity f is small everywhere on the boundary except near the real axis then it is small in tilde Q} applies to give us that $f(z) \leq 2^{1+2s}M(1+2^{1-s}) (\frac{b}{a/2})^s+m$ for all $z\in\tilde{Q}(a/2)$. 
	
	Put succinctly, we find $f(z) \leq C (b^s + b^{-s}e^{-b d})$ for all $z\in\tilde{Q}(a/2)$ for some constant $C>0$ (different constant than before), $0<s<1$ and $d\gg1$, and we are free to choose $b<a/2$. Our goal is to get decay in $d$. If we pick $b:=d^{-s}$ for example we get the desired decay in $d$.
	
\end{proof}
\end{prop}

Finally we are ready to get the \emph{exponential} decay of the infinite volume Green's function, which concludes the proof for the first part of \cref{thm:localization}.

\begin{lem}
	\label{lem:any decay implies exponential decay}
	For fixed $z$, uniformly in $|\Im\{ z\} |$:
	assume that $\mathbb{E}[\norm{G_{[1,n]}(1,n;\,z)}^{s}]\to0$
	as $|n|\to\infty$ for some $s\in(0,1)$.
	Then $\mathbb{E}[\norm{G(1,n;\,z)}^{s'}]\leq Ce^{-\mu|n|}$
	for all $n\in\mathbb{Z}$ sufficiently large, for some $s'\in(0,1)$
	sufficiently small, $\mu>0$.
	\begin{proof}
		We have $H_{[x,y]}=\sum_{x'=x+1}^{y}\mathcal{T}_{x'}$ with $(\mathcal{T}_{j}\psi)(x):=\delta_{j,x}T_j\psi_{j-1}+\delta_{j-1,x}T_j^\ast\psi_{j}$ as
		the finite volume Dirichlet restriction of $H$ onto $[x,y]\cap\mathbb{Z}$,
		$x<y$. Then the resolvent equation yields 
		\begin{align*}
			R & =  R_{(-\infty,y]}-\sum_{x'=y+1}^{\infty}R_{(-\infty,y]}\mathcal{T}_{x'}R\,.
		\end{align*}
		Taking the $(x,y)$ matrix element yields (suppressing
		the $z$ variable for the moment) 
		\begin{align*}
			G(x,y) & =  G_{(-\infty,y]}(x,y)-\sum_{x'=y+1}^{\infty}G_{(-\infty,y]}(x,x')T_{x'}G(x'-1,y)+G_{(-\infty,y]}(x,x'-1)T_{x'}^{\ast}G(x',y)\\
			& =  G_{(-\infty,y]}(x,y)(\mathds{1}-T_{y+1}^{\ast}G(y+1,y))\,,
		\end{align*}
		where the second line follows because the matrix elements of $R_{(-\infty,y]}$
		outside of $(-\infty,y]$ are zero. Next we have again
		by the resolvent equation 
		\begin{align*}
			R_{(-\infty,y]} & =  R_{[x,y]}-\sum_{x'=-\infty}^{x}R_{(-\infty,y]}\mathcal{T}_{x'}R_{[x,y]}\,,
		\end{align*}
		so that taking the $(x,y)$ matrix element we get
		\begin{align*}
			G_{(-\infty,y]}(x,y) & =  G_{[x,y]}(x,y)-\\
			& - \sum_{x'=-\infty}^{x}G_{(-\infty,y]}(x,x')T_{x'}G_{[x,y]}(x'-1,y)+G_{(-\infty,y]}(x,x'-1)T_{x'}^{\ast}G_{[x,y]}(x',y)\\
			& =  (\mathds{1}-G_{(-\infty,y]}(x,x-1)T_{x}^{\ast})G_{[x,y]}(x,y)\,,
		\end{align*}
		where the second line follows again because the matrix elements of
		$R_{[x,y]}$ outside of $[x,y]$ are zero.
		So we find 
		\begin{align*}
			\mathbb{E}[\norm{G(x,y)}^{s}] & \leq  C\mathbb{E}[\norm{G_{[x,y]}(x,y)}^{s'}]\,,
		\end{align*}
		where $s'=4s$ for example. To get $C$ one has to invoke the H\"older
		inequality twice as well as the a-priori bound which is known for
		$\mathbb{E}[\norm{G(x,x+1;\,z)}^{s}]$ uniformly
		in $z$.
		
		The upshot is that we may concentrate on exponential decay of $g(n):=\mathbb{E}[\norm{G_{[1,n]}(1,n)}^{s}]$
		in $n$ (by stationarity it does not matter to shift the object by
		$x-1$ and call $y-x+1=:n$).
		
		Our next procedure is to get a one step bound between $g(n+m)$
		and $g(n)g(m)$ for any $n$, $m$:
		
		We use again the resolvent identity to get 
		\begin{align*}
			G_{[1,n+m]}(1,n+m) & =  -G_{[1,n+m]}(1,n)T_{n+1}^{\ast}G_{[n+1,n+m]}(n+1,n+m)\,,
		\end{align*}
		(note $G_{[n+1,n+m]}(1,n+m)=0$) and 
		\begin{align*}
			G_{[1,n+m]}(1,n) & =  G_{[1,n]}(1,n)-G_{[1,n]}(1,n)T_{n+1}^{\ast}G_{[1,n+m]}(n+1,n)\,,
		\end{align*}
		so that 
		\begin{align*}
			G_{[1,n+m]}(1,n+m) & =  -G_{[1,n]}(1,n)T_{n+1}^{\ast}G_{[n+1,n+m]}(n+1,n+m)+\\
			&   +G_{[1,n]}(1,n)T_{n+1}^{\ast}G_{[1,n+m]}(n+1,n)T_{n+1}^{\ast}G_{[n+1,n+m]}(n+1,n+m)\,.
		\end{align*}
		Taking the fractional moments expectation value, using the triangle
		inequality as well as the submultiplicativity of the norm, we find
		\begin{align*}
			\mathbb{E}[\norm{G_{[1,n+m]}(1,n+m)}^{s}] & \leq  \mathbb{E}[\norm{G_{[1,n]}(1,n)}^{s}\norm{T_{n+1}^{\ast}}^{s}\norm{G_{[n+1,n+m]}(n+1,n+m)}^{s}]+\\
			&   +\mathbb{E}[\norm{G_{[1,n]}(1,n)}^{s}\norm{T_{n+1}^{\ast}}^{2s}\norm{G_{[1,n+m]}(n+1,n)}^{s}\times\\
			&\times\norm{G_{[n+1,n+m]}(n+1,n+m)}^{s}]\,.
		\end{align*}
		Note that in the first line, the first and last factors in the expectation
		are actually independent of each other and both independent of $T_{n+1}$.
		Hence that expectation factorizes. In the second line, again the first
		and last factors do not depend on $T_{n+1}$ (yet the middle one does)
		so that we can perform the integration over $T_{n+1}$ \emph{first,
		}which would involve integration only over $\norm{T_{n+1}^{\ast}}^{2s}\norm{G_{[1,n+m]}(n+1,n)}^{s}$.
		We use H\"older once and the a-priori bound on $G(n+1,n)$,
		which requires only integration over $T_{n+1}$. After the bound on
		that integral the remaining integral factorizes as the two remaining
		factors are independent of each other. We find 
		\begin{align*}
			\mathbb{E}[\norm{G_{[1,n+m]}(1,n+m)}^{s}] & \leq  C\mathbb{E}[\norm{G_{[1,n]}(1,n)}^{s}]\mathbb{E}[\norm{G_{[n+1,n+m]}(n+1,n+m)}^{s}]
		\end{align*}
		with $C:=\mathbb{E}[\norm{T_{n+1}^{\ast}}^{s}]+\mathbb{E}[\norm{T_{n+1}^{\ast}}^{4s}]^{\frac{1}{2}}\mathbb{E}[\norm{G_{[1,n+m]}(n+1,n)}^{2s}]^{\frac{1}{2}}$.
		If $s<1$ is sufficiently small and we assume that there are moments
		for $\norm{T_{n+1}^{\ast}}^{s}$ for such $s$, then we find $0<C<\infty$. 
		
		The crucial point now is that due to stationarity, 
		\begin{align*}
			\mathbb{E}[\norm{G_{[1,m]}(1,m)}^{s}] & =  \mathbb{E}[\norm{G_{[n+1,n+m]}(n+1,n+m)}^{s}]\,.
		\end{align*}
		The final result is that 
		\[
			g(n+m)\leq Cg(n)g(m)\quad\forall n\in\mathbb{N}\,.
		\]
		Now we use the assumption that $g$ is decaying, which means we could
		find some $n_{0}\in\mathbb{N}$ sufficiently large so that $\beta:=Cg(n_{0})<1$.
		Then for any $n\in\mathbb{N}$, write $n=pn_{0}+q$ (for some $p\in\mathbb{N}$,
		$q\in\mathbb{N}$ with $0\leq q<n_{0}$). We find by iteration, defining
		$\mu:=-\log(\beta)>0$: 
		\begin{align*}
			g(n) & =  g(pn_{0}+q) \leq  Cg(pn_{0})g(q)\\
			& \leq  (Cg(n_{0}))^{p}g(q)\\&   (g\text{ is decaying and }q<n_{0})\\
			& \leq  C'\beta^{p} =  C'\exp(-\mu\frac{n-q}{n_{0}}) \leq  C''\exp(-\mu'n)\,.
		\end{align*}
		
	\end{proof}
\end{lem}

\section{Localization at zero energy}
\label{sec:zero energy}
The foregoing discussion only worked at non-zero energies. There were
two reasons for that:
\begin{enumerate}
\item We could not guarantee that the zero energy Lyapunov spectrum has a gap. This goes back to \cref{prop:open subset of symplectic group}.
\item We could not get an a-priori bound on the diagonal matrix element
of the Green's function $G(x,x;\,z)$ which is uniform
as $z\to0$. This goes back to \cref{cor:a-priori bound for xx green's function}.
\end{enumerate}
In order to deal with that special situation, we have to consider
the Schr\"odinger equation at zero energy and then conclude about slightly
non-zero values of the energy. We note that $H$ is not invertible, so an expression for
the resolvent like $R(0)\equiv H^{-1}$ does not make sense. Hence we use the finite-volume regularization in this section.
\subsection{Finite volume localization}
Thus we are considering the operator $H_{[1,L]}$ for
some $L\in\mathbb{N}$, which is just $H$ with Dirichlet boundary
conditions. It is the finite $L\times L$ ``band'' matrix of $N\times N$
blocks given as 
\begin{align*}
H_{[1,L]} & =  \begin{pmatrix}0 & T_{2}^{\ast}\\
T_{2} & 0 & T_{3}^{\ast}\\
 & T_{3} &    & \ddots\\
 &    & \ddots & T_{L}^{\ast}\\
 &   \ddots & T_{L} & 0
\end{pmatrix}\,.
\end{align*}
\begin{prop}
$H_{[1,2n+1]}$ is not invertible and $H_{[1,2n]}$
is invertible, for all $n\in\mathbb{N}$.
\begin{proof}
Using the left boundary condition we have $\psi_{0}=0$. Then the
Schr\"odinger equation at zero energy implies that the wave function
at all even sites is zero, by iteration: 
\begin{align*}
T_{2}^{\ast}\psi_{2}+T_{1}\psi_{0} & =  0
\end{align*}
and so on. Thus the even sites are all zero by the left boundary condition.

If we consider $H_{[1,2n]}$, then the right boundary
condition is $\psi_{2n+1}=0$, and then again by using the Schr\"odinger
equation the wave function at all odd sites must be zero. Hence $H_{[1,2n]}\psi=0$
implies $\psi=0$, that is, $H_{[1,2n]}$ is invertible.

If on the other hand we have $H_{[1,2n+1]}$, then the
right boundary condition is $\psi_{2n+2}=0$, which doesn't give any
new information: it is merely compatible with having the wave function
at all even sites zero. Hence, the wave function at odd sites is unconstrained.
Once $\psi_{1}\in\mathbb{C}^{N}$ is chosen, we use the equation to
obtain the wave function's value at odd sites along the entire chain:
\begin{align*}
T_{3}^{\ast}\psi_{3}+T_{2}\psi_{1} & =  0
\end{align*}
and so on. Hence $\ker(H_{[1,2n+1]})  \cong  \mathbb{C}^{N}$
and $\ker(H_{[1,2n]}) =  \{0\}$.
\end{proof}
\end{prop}

Consequently, it would not make sense to consider the resolvent for
odd chain-lengths, and we shall restrict our attention to $H_{[1,2n]}$.

It turns out that it is easy to calculate the matrix elements of $R_{[1,2n]}(0)\equiv(H_{[1,2n]})^{-1}$.
We only need to describe the elements of $R_{[1,2n]}(0)$
on the diagonal and above it due to the self-adjointness of $H_{[1,2n]}$. 
\begin{prop}
The only non-zero matrix elements $R_{[1,2n]}(0)$ on or above the diagonal are given by 
\begin{align*}
G_{[1,2n]}(2k,2l+1;\,0) & =  (-T_{2k}^\circ T_{2k-1})\dots(-T_{2l+4}^\circ T_{2l+3})T_{2l+2}^\circ
\end{align*}
for all $(k,l)\in\mathbb{N}^{2}$ such that $(2k,2l+1)\in[1,2n]^{2}$
and such that $k>l$.
\begin{proof}
Let $l\in\mathbb{N}$ be given such that $2l+1\in[1,2n]$.
We start from the left boundary condition, which is that $G_{[1,2n]}(0,2l+1;\,z)\equiv0$.
We then evaluate the Schr\"odinger equation at zero energy in the left
most position to find 
\begin{align*}
T_{2}^{\ast}G_{[1,2n]}(2,2l+1;\,0) & =  \mathds{1}\delta_{1,2l+1}\,.
\end{align*}
If $l=0$ then we find $G_{[1,2n]}(2,1;\,0)=T_{2}^\circ$.
Otherwise $G_{[1,2n]}(2,2l+1;\,0)=0$. We
continue in this fashion to find that $G_{[1,2n]}(2k,2l+1;\,0)=0$
as long as $k\leq l$ and the equation above once $k=l+1$, and then
iterate for $k>l+1$.

We cannot proceed in the same way for $G_{[1,2n]}(1,2l+1;\,0)$
because the boundary condition on the left doesn't say anything about
it. Instead we must use the boundary condition on the right, which
says $G_{[1,2n]}(2n+1,2l+1;\,z)\equiv0$.
In the same way we use the Schr\"odinger equation to conclude about
$G_{[1,2n]}(2n-1,2l+1;\,0)=0$: 
\begin{align*}
T_{2n+1}^{\ast}\underbrace{G_{[1,2n]}(2n+1,2l+1;\,0)}_{\equiv0}+T_{2n}G_{[1,2n]}(2n-1,2l+1;\,0) & =  \mathds{1}\underbrace{\delta_{2n,2l+1}}_{=0}\,.
\end{align*}
Since the right hand side will \emph{always} be zero (due to the difference
in parity), we find that 
\begin{align*}
G_{[1,2n]}(2k+1,2l+1;\,0) & =  0
\end{align*}
for all $k$ such that $2k+1\in[1,2n]$. 

In a similar way we also find that 
\begin{align*}
G_{[1,2n]}(2k,2l;\,0) & =  0
\end{align*}
 for all $2k$ and $2l$ within the chain, using the boundary condition
on the left and then evolving to the right.
\end{proof}
\end{prop}

Now that we know that all diagonal matrix elements of $R_{[1,2n]}(0)$
are zero, we proceed to get an expression at non-zero energy, but
still finite volume:
\begin{prop}
\label{prop:Uniform bound on diagonal matrix element of finite volume Green's function}If
the Lyapunov exponents are all non-zero for $z=0$ (so we assume more than what \cref{cor:The Lyapunov spectrum is simple}
automatically gives), then we have 
\begin{align*}
\mathbb{E}[\norm{G_{[1,2n]}(x,x;\,z)}^{s}] & <  C
\end{align*}
for some constant uniformly in $z$ (as $z\to0$), uniformly in $n$,
and independent of $x$.
\begin{proof}
We use the resolvent identity to get 
\begin{align*}
G_{[1,2n]}(x,x;\,z) & =  G_{[1,2n]}(x,x;\,z)-G_{[1,2n]}(x,x;\,0)\\
 & =  \left\langle \delta_{x},[R_{[1,2n]}(z)-R_{[1,2n]}(0)]\delta_{x}\right\rangle \\
 &   (\text{Resolvent identity})\\
 & =  \left\langle \delta_{x},zR_{[1,2n]}(z)R_{[1,2n]}(0)\delta_{x}\right\rangle \\
 & =  \sum_{y=1}^{2n}zG_{[1,2n]}(x,y;\,z)G_{[1,2n]}(y,x;\,0)\,.
\end{align*}
So 
\begin{align*}
\mathbb{E}[\norm{G_{[1,2n]}(x,x;\,z)}^{s}] & \leq  \sum_{y=1}^{2n}|z|^{s}(\mathbb{E}[\norm{G_{[1,2n]}(x,y;\,z)}^{2s}])^{\frac{1}{2}}(\mathbb{E}[\norm{G_{[1,2n]}(y,x;\,0)}^{2s}])^{\frac{1}{2}}\,.
\end{align*}
We now use \cref{lem:any decay implies exponential decay} (namely that the finite volume complex energy Green's function is exponentially decaying) to conclude: 
\begin{align*}
\mathbb{E}[\norm{G_{[1,2n]}(x,x;\,z)}^{s}] & \leq  \sum_{y=1}^{2n}|z|^{s}(|z|^{-s}e^{-\mu_{s}|x-y|})(e^{-\mu_{s}'|x-y|})\\
 & <  C\,.
\end{align*}
We note that the bound on $\mathbb{E}[\norm{G_{[1,2n]}(y,x;\,0)}^{2s}]$
does not include a factor of $|z|^{-s}$ precisely because
we know that at zero energy $G_{[1,2n]}(y,y;\,0)=0$.
\end{proof}
\end{prop}
\subsection{Infinite volume localization}
Our next goal is to conclude the same bound for the infinite system.
By ergodicity instead of working with $[1,2n]$ we could
just as well work with $[-n+1,n]$, which also holds
an even number of sites.
\begin{prop}
We have 
\begin{align*}
\slim_{n\to\infty}R_{[-n+1,n]}(z) & =  R(z)
\end{align*}
for all fixed $z\in\mathbb{C}\backslash\mathbb{R}$.
\begin{proof}
The operator $H_{[-n+1,n]}$ is defined as 
\begin{align*}
H_{[-n+1,n]} & =  \sum_{j=-n+2}^{n}\mathcal{T}_{j}
\end{align*}
with $(\mathcal{T}_{j}\psi)(x):=\delta_{j,x}T_j\psi_{j-1}+\delta_{j-1,x}T_j^\ast\psi_{j}$ so that the resolvent identity gives 
\begin{align*}
R_{[-n+1,n]} & =  R+R(H-H_{[-n+1,n]})R_{[-n+1,n]}\\
 & =  R+\sum_{j\in\mathbb{Z}\setminus\{-n+2,\dots,n\}}R\mathcal{T}_{j}R_{[-n+1,n]}\,.
\end{align*}
Now
\begin{align*}
\|\sum_{j\in\mathbb{Z}\backslash\{-n+2,\dots,n\}}\mathcal{T}_{j}\psi\|^2&=\sum_{l\in\mathbb{Z}}\|\sum_{j\in\mathbb{Z}\backslash\{-n+2,\dots,n\}}\left\langle \delta_{l},\mathcal{T}_{j}\psi\right\rangle \|^2\\
&=\sum_{j\in\mathbb{Z}}\|\sum_{l\in\mathbb{Z}\backslash\{-n+2,\dots,n\}}(\delta_{j,l}T_{l}\psi_{l-1}+\delta_{j,l-1}T_{l}^{\ast}\psi_{l})\|^{2}\\
&\leq\sum_{j\in\mathbb{Z}\backslash\{-n+2,\dots,n\}}\|T_{j}\psi_{j-1}\|^{2}+\sum_{j\in\mathbb{Z}\backslash\{-n+1,\dots,n-1\}}\|T_{j+1}^{\ast}\psi_{j+1}\|^{2}\,.
\end{align*}
Next we observe that
\begin{align*}
(R_{[-n+1,n]}(z)\psi)_{j} & =  -z^{-1}\psi_{j}\quad\forall j\in\mathbb{Z}\backslash\{-n+2,\dots,n\}
\end{align*}
so that 
\begin{align*}
\|R(z)\sum_{j=-n+2}^{n}\mathcal{T}_{j}R_{[-n+1,n]}(z)\psi\|^{2} & \leq  \norm{R(z)}^{2}|z|^{-2}\sum_{j\in\mathbb{Z}\backslash\{-n+2,\dots,n\}}\norm{T_{j}\psi_{j-1}}^{2}+\\
 &   +\norm{R(z)}^{2}|z|^{-2}\sum_{j\in\mathbb{Z}\backslash\{-n+1,\dots,n-1\}}\norm{T_{j+1}^{\ast}\psi_{j+1}}^{2}\,.
\end{align*}
But $\psi\in l^{2}$ and $\norm{R(z)}\leq|\Im\{z\} |^{-1}$,
so that the right hand side converges to zero as $n\to\infty$, all
at fixed $z\neq0$.
\end{proof}
\end{prop}

\begin{cor}
We have 
\begin{align*}
\lim_{n\to\infty}\norm{G_{[-n+1,n]}(x,x;\,z)} & =  \norm{G(x,x;\,z)}
\end{align*}
for all $z\in\mathbb{C}\backslash\mathbb{R}$, so that using Fatou's
lemma, 
\begin{align*}
\mathbb{E}[\norm{G(x,x;\,z)}^{s}] & \leq  \lim_{n\to\infty}\mathbb{E}[\norm{G_{[-n+1,n]}(x,x;\,z)}^{s}]\,.
\end{align*}
But since the bound \cref{prop:Uniform bound on diagonal matrix element of finite volume Green's function}
is uniform in $n$, we find that $\mathbb{E}[\norm{G(x,x;\,z)}^{s}]$
is bounded uniformly in $|\Im\{z\} |$ under
the same assumptions on the Lyapunov spectrum as in \cref{prop:Uniform bound on diagonal matrix element of finite volume Green's function}.

As a result, we may now go back to the previous section and apply all the proofs there, extending them so that it holds uniformly including in the limit $z\to0$ as long as the Lyapunov spectrum doesn't include zero, at \emph{all}
real energies, including zero energy. This concludes the proof of the second part of \cref{thm:localization}. \end{cor}

\section{Chiral random band matrices and the \texorpdfstring{$\sqrt{L}$}{sqrt-L} conjecture}\label{sec:chiral RBM}
In this section we remark briefly about random band matrices from the point of view of our chiral model. To avoid confusion, now the size of our random matrices is $W$ instead of $N$ as in the rest of the paper. Random band matrices are random Hermitian matrices $A$ of size\footnote{To be compatible with the rest of the paper we denote a size of a matrix here by $L$ instead of the usual $N$ in the literature, the usual name of the conjecture is "the $\sqrt{N}$ conjecture".} $L\times L$, for some $L\in\NN$, such that there is a parameter $W\in\NN_{<L/2}$ with $A_{ij} = 0$ for all $|i-j|>W$. I.e., they are of finite band width $W$. The random distribution that is usually considered is that where each entry within the non-zero band is independent and identically distributed (except the Hermitian constraint), say, as a Gaussian centered at zero (though this detail is largely unimportant for the relevant questions to follow). These models are interesting due to the fact that they exhibit a phase transition, from localization for $W\ll \sqrt{L}$ to delocalization $W\gg \sqrt{L}$. I.e., a transition is conjectured for $W \sim \sqrt{L}$ \cite{PhysRevLett.64.1851,PhysRevE.48.R1613}. Indeed, the case where $W$ is a constant (w.r.t. $L\to\infty$) is simply the Anderson model on a strip (of width $W$), which is completely localized, and the case $W=L/2$ is a random matrix from the Gaussian unitary ensemble, known to be delocalized. Both phases have been established mathematically (see \cite{Schenker2009,10.1093/imrn/rnx145} for the localized regime for $W\sim L^{1/7}$ and e.g. \cite{Bourgade2017,https://doi.org/10.1002/cpa.21895} and references therein for the delocalized regime for $W\sim L^{3/4}$) but the precise transition at $W\sim\sqrt{L}$ has yet to have been established (but see e.g. \cite{Shcherbina2014,Shcherbina2017} for some recent promising progress).

Using finite rank perturbation theory and a-priori bounds, the question of localization for random band matrices is equivalent to the question of the localization length for the Anderson model on the strip, or the following generalization of it: \begin{align} (H\psi)_n = T_{n+1}^\ast \psi_{n+1}+T_n\psi_{n-1}+V_n\psi_n \qquad (\psi\in\ell^2(\mathbb{Z})\otimes\CC^W;n\in\ZZ)\label{eq:Wegner N orbital model}\end{align} where $\Set{V_n}_n$ is a random i.i.d. sequence chosen from $\mathrm{GUE}(W)$ and $\Set{T_n}_n$ is a random i.i.d. sequence chosen from $\mathrm{Ginibre}(W)$. Strictly speaking this model is called the Wegner-$W$ orbital model \cite{10.1093/imrn/rnx145}, and to get the random band model one should take $T_n$ to be distributed as random upper triangular. The distinction, however, should not matter for our purposes. 
Then the $\sqrt{L}$ conjecture from the localization side is equivalent to the model \cref{eq:Wegner N orbital model} exhibiting localization length of order $W$ (measured in distance between slices of width $W$); one way to formulate that is to establish the FM condition for \cref{eq:Wegner N orbital model} with decay rate $W$: 
\begin{align} \sup_{\eta\neq0}\mathbb{E}\left[\norm{G(x,y;E+\ii\eta)}^s\right]\leq C\exp(-\mu|x-y|/W)\qquad(x,y\in\mathbb{Z}) \end{align} for some $s\in(0,1)$ and $0<C,\mu<\infty$ where $\mu\sim \mathcal{O}(1)$ in $W\to\infty$. This should hold for any $E$ in the almost-sure spectrum, i.e., there is no expectation for a phase transition in $E$. Thus, the preceding sections which connect localization length with the smallest Lyapunv exponent lead us to re-formulate the $\sqrt{L}$ conjecture (at energy $\lambda$) from the localization side as: \begin{align} \inf_{j\in\Set{1,\dots,2W}}|\gamma_j(\lambda)| \gtrsim \frac{1}{W}\,. \label{eq:sqrt-N conjecture in terms of Lyapunv spectrum}\end{align} 

It turns out--and this is the point of this section--that \emph{for \cref{eq:Hamiltonian} (i.e. $V_n=0$ in \cref{eq:Wegner N orbital model}) at zero energy}, the Lyapunov spectrum can be calculated exactly for special choices of $\alpha_0,\alpha_1$.

\begin{thm}
	Consider the model \cref{eq:Hamiltonian} with distributions $\alpha_0,\alpha_1$ both Ginibre with variances $\sigma_0,\sigma_1$ respectively. If $\sigma_0=\sigma_1$ then all zero energy Lyapunov exponents are zero. If, on the other hand, one chooses e.g. $\sigma_0 = \exp(-1/W)$ and $\sigma_1 = \exp(-2/W)$, the $\sqrt{W}$ conjecture holds at zero energy, i.e., \begin{align} \inf_{j\in\Set{1,\dots,2W}}|\gamma_j(0)| \sim \frac{1}{W}\,. \label{eq:sqrt N conjecture expressed via LE}\end{align}
\end{thm} 
\begin{proof}
	At zero energy, the Lyapunov spectrum breaks apart as in \cref{eq:zero energy transfer matrices}, so it is enough to study the smallest Lyapunov exponent of the sequence of $W\times W$ matrices $\Set{A_n B_n^{\circ}}_n$ where $\Set{A_n}_n$ are distributed according to $\alpha_0$ and $\Set{B_n}_n$ according to $\alpha_1$. We recall that the Ginibre distribution is given by $$ \dif{\mathrm{Ginibre}_\sigma(A)} = (\pi\sigma^2)^{-W^2}\exp(-\frac{1}{\sigma^2}\tr(|A|^2))\dif{A} $$ where $\dif{A}$ denotes the Lebesgue measure on $\Mat_{W}(\CC)\cong\CC^{W^2}$. Hence, the distribution $\mathbb{P}_{\mathrm{CBZE}}$ of each element of the sequence $\Set{A_n B_n^{\circ}}_n$ is given by $$ \mathbb{E}_{\mathrm{CBZE}}\left[f\right] = \int_{A,B\in\Mat_{W}(\CC)}(\pi^2\sigma_0^2\sigma_1^2)^{-W^2}\exp\left(-\frac{1}{\sigma_0^2}\tr(|A|^2)-\frac{1}{\sigma_1^2}\tr(|B|^2)\right)f(AB^\circ)\dif{A}\dif{B} $$ for any measurable $f:\Mat_{W}(\CC)\to\CC$. Here the acronym $\mathrm{CBZE}$ stands for the distribution of "Chiral Band matrices at Zero Energy". 
	
	In particular, we see that: (1) $\mathbb{P}_{\mathrm{CBZE}}\left[\mathrm{GL}_W(\CC)^c\right]=0$, and (2) for any fixed unitary $W\times W$ matrix $U$, $|AB^\circ|^2$ has the same distribution as $|AB^\circ U|^2$. Indeed, this is clear from the cyclicity of the trace, $U$ being unitary so its Jacobian is the identity, and $$U^\ast (AB^\circ)^\ast AB^\circ U = (U^\ast B U)^{-1}(U^\ast A^\ast A U)(U^\ast B U)^\circ \longrightarrow(AB^\circ)^\ast A B^\circ \,.$$ These two conditions thus allow us to apply the same arguments as in \cite{Newman_1986_cmp/1104114627} to conclude that the Lyapunov exponents $\Set{\xi_j}_{j=1}^W$ of the sequence $\Set{A_n B_n^{\circ}}_n$ are given by $$ \xi_1+\dots+\xi_k = \mathbb{E}_{\mathrm{CBZE}}\left[\log\left(\norm{AB^\circ e_1\wedge\dots\wedge e_k}\right)\right]\,. $$ But following the arguments in \cite{Newman_1986_cmp/1104114627}, even more is true. Since the choice of the basis $\Set{e_j}_j$ does not matter for the calculation of the exponents, we could just as well use $\Set{B^\ast e_j}_j$ instead so we also have $$ \xi_1+\dots+\xi_k = \mathbb{E}_{\mathrm{CBZE}}\left[\log\left(\frac{\norm{AB^\circ B^\ast e_1\wedge\dots\wedge e_k}}{\norm{ B^\ast e_1\wedge\dots\wedge e_k}}\right)\right] $$ so we see that the task reduces to calculation of the exponents $\Set{\xi^A_j}_{j=1}^W$ of $A$ and $\Set{\xi^B_j}_{j=1}^W$ of $B$ separately, and that $\xi_j = \xi_j^A - \xi_j^B$ (the distribution of $B$ is invariant under adjoints). Since both $A$ and $B$ are Ginibre, the very same calculation in \cite{Newman_1986_cmp/1104114627} applies and we obtain $$\xi_j = \log(\sigma_0/\sigma_1)\qquad(j=1,\dots,W)$$ so all exponents are degenerate. The two statements in the theorem readily follow.
\end{proof}

\bigskip
\bigskip
\noindent\textbf{Acknowledgements:} 
I am gratefully indebted to Gian Michele Graf for many stimulating discussions. I also wish to thank Michael Aizenman, Peter M\"uller, Hermann Schulz-Baldes and G\"unter Stolz. I acknowledge support by the Swiss National Science Foundation (grant number P2EZP2\_184228).
\bigskip

\section{Appendix}
\subsection{Deterministic bounds on the Fermi projection}

Our goal in this subsection is to prove \cref{cor:almost sure consequences of localization for the Fermi projection}. Using \cite[Theorem 2]{Aizenman_Graf_1998}, one obtains that if $P := \chi_{(-\infty,E)}(H)$ where $E\in\mathbb{R}$ is an energy for which \cref{eq:fractional moment condition} holds (either automatically for all $E\neq0$ by the first part of \cref{thm:localization} or with further assumptions by its second part) then \begin{align}
\mathbb{E}[\|P(x,y)\|]\leq C e^{-\mu|x-y|}
\end{align}
for some (deterministic) constants $C>0$ and $\mu>0$. The following claim proves \cref{cor:almost sure consequences of localization for the Fermi projection} then.

\begin{claim}
Let $A:\Omega\to\mathcal{B}(\mathcal{H})$ be an ergodic random operator (with the usual assumptions on the disorder probability space $(\Omega,\mathcal{A},\mathbb{P},T)$ with $T$ being the ergodic action of $\mathbb{Z}$). If $$\mathbb{E}[\|A(x,y)\|]\leq C e^{-\mu|x-y|}$$ for some (deterministic) constants  $C,\mu>0$, then almost surely for some (random) constants $C'>0$ and deterministic constants $\mu',\nu>0$, $$ \sum_{x,y\in\mathbb{Z}} \|A(x,y)\|(1+|x|)^{-\nu}e^{\mu'|x-y|} \leq C' <\infty\,. $$ 
%
\end{claim}

The proof of this claim has appeared already in many places, e.g., \cite[Prop. A.1]{Shapiro2019}.

\subsection{The Hermitian symplectic group}\label{sec:The Hermitian symplectic group}
Since the Hermitian symplectic group is not a very canonical object in the mathematics literature (though it does appear, e.g., in \cite{Sadel2010}), we spell out its basic properties below for convenience.
\begin{defn}
	\label{def:The Hermitian Symplectic Group}The (Hermitian) symplectic
	group $\HSG$ is defined as 
	\begin{align*}
		\HSG & \equiv  \{M\in \Mat_{2N\times2N}(\mathbb{C})|M^{\ast} J M= J\}
	\end{align*}
	with $ J\equiv\begin{pmatrix}0 & -\mathds{1}_{N}\\
		\mathds{1}_{N} & 0
	\end{pmatrix}$ the standard symplectic form. Here we abbreviate $G:=\HSG$.
	In this section we write $M$ in $N$-block form as $M=\begin{pmatrix}A & B\\
		C & D
	\end{pmatrix}$, where $A$, $B$, $C$, and $D$ are unrelated to their meanings
	outside of this section. 
\end{defn}

\begin{prop}
	$G$ is a group under matrix multiplication.
	\begin{proof}
		Taking the determinant of 
		\begin{align*}
			M^{\ast} J M & =   J
		\end{align*}
		shows that $|\det(M)|=1$ so that $M$ is invertible
		if it is in $G$. From the defining relation we have
		\begin{align*}
			M & =  (M^{\ast} J)^{-1} J =  M^{-1}A^{\ast-1} J
		\end{align*}
		so that using the commutativity of adjoint and inverse, we get:
		\begin{align*}
			M^{-1} & =   J^{-1}M^{\ast} J =   J^{-1}((M^{-1})^{-1})^{\ast} J =   J^{-1}((M^{-1})^{\ast})^{-1} J =  ((M^{-1})^{\ast} J)^{-1} J\,,
		\end{align*}
		so that 
		\begin{align}
			(M^{-1})^{\ast} J M^{-1} & =   J\label{eq:inverse of conjugate symplectic is conjugate symplectic}
		\end{align}
		and we find that $M^{-1}\in G$ as well. Hence the inverse map restricted
		to $G$ lands in $G$. If 
		\begin{align*}
			M_{i}^{\ast} J M_{i} & =   J
		\end{align*}
		for $i\in\{1,2\}$ then 
		\begin{align*}
			M_{2}^{\ast}M_{1}^{\ast} J M_{1}M_{2} & =  M_{2}^{\ast} J M_{2} =   J\,.
		\end{align*}
		Hence matrix multiplication map restricted to $G^{2}$ lands in $G$.
		
		Associativity is inherited by the associativity of general matrix
		multiplication. Finally, the identity matrix is of course conjugate
		symplectic.
	\end{proof}
\end{prop}

\begin{rem}
	Note that $G$ is unitarily equivalent to $U(N,N)$,
	the indefinite unitary group. 
\end{rem}

\begin{prop}
	$G$ is closed under adjoint.
	\begin{proof}
		Let $M\in G$ be given. Then $M^{\ast} J M= J$. We want to
		show that $M J M^{\ast}= J$. Taking the inverse of \cref{eq:inverse of conjugate symplectic is conjugate symplectic}
		we find 
		\begin{align*}
			((M^{-1})^{\ast} J M^{-1})^{-1} & =   J^{-1}\,.
		\end{align*}
		Use the fact that $ J^{-1}=- J$ and again the commutativity
		of adjoint and inverse maps to find: 
		\begin{align*}
			M J M^{\ast} & =   J\,,
		\end{align*}
		so that $M^{\ast}\in G$ as desired.
	\end{proof}
\end{prop}

\begin{prop}
	\label{prop:general block description of Hermitian symplectic group}
	$G$ may be described as 
	\begin{align}
		G & =  \left\{\left.\begin{pmatrix}A & B\\
			C & D
		\end{pmatrix}\in \Mat_{2N}(\mathbb{C})\right|\begin{array}{c}
			(A,B,C,D)\in \Mat_{N}(\mathbb{C})^{4}:\\
			A^{\ast}C\text{ and }B^{\ast}D\text{ are S.A.}\\
			\text{ and }A^{\ast}D-C^{\ast}B=\mathds{1}_{N}
		\end{array}\right\}\label{eq:Description of the symplectic group}\,.
	\end{align}
	\begin{proof}
		Starting from $M^\ast J M \stackrel{!}{=} J$ we have
		\begin{align*}
			\begin{pmatrix}A & B\\
				C & D
			\end{pmatrix}^{\ast}\begin{pmatrix}0 & -\mathds{1}_{N}\\
				\mathds{1}_{N} & 0
			\end{pmatrix}\begin{pmatrix}A & B\\
				C & D
			\end{pmatrix} & =  \begin{pmatrix}A^{\ast} & C^{\ast}\\
				B^{\ast} & D^{\ast}
			\end{pmatrix}\begin{pmatrix}0 & -\mathds{1}_{N}\\
				\mathds{1}_{N} & 0
			\end{pmatrix}\begin{pmatrix}A & B\\
				C & D
			\end{pmatrix}\\
			& =  \begin{pmatrix}A^{\ast} & C^{\ast}\\
				B^{\ast} & D^{\ast}
			\end{pmatrix}\begin{pmatrix}-C & -D\\
				A & B
			\end{pmatrix}\\
			& =  \begin{pmatrix}-A^{\ast}C+C^{\ast}A & -A^{\ast}D+C^{\ast}B\\
				-B^{\ast}C+D^{\ast}A & -B^{\ast}D+D^{\ast}B
			\end{pmatrix}\\
			& \stackrel{!}{=}  \begin{pmatrix}0 & -\mathds{1}_{N}\\
				\mathds{1}_{N} & 0
			\end{pmatrix}\,.
		\end{align*}
	\end{proof}
\end{prop}



\begin{rem}
	\label{note:If D is invertible then a symplectic matrix is uniquely determined by just three entries}By
	the last equation in \cref{eq:Description of the symplectic group},
	if $D\in \GL_{N}(\mathbb{C})$, then $M$ is specified by
	only three block entries $B$, $C$ and $D$ and $A$ may be solved
	for to satisfy the symplectic constraint.
\end{rem}

\begin{proof}[Proof of \cref{lem:General form of Hermitian symplectic matrix when D is invertible}]
	\cref{eq:form for A block in a Hermitian symplectic matrix with D invertible} is equivalent to the last condition of \cref{eq:Description of the symplectic group} when $D$ is invertible. So by \cref{prop:general block description of Hermitian symplectic group} we need only verify that $A^\ast C$ and $B^\ast D$ are self-adjoint. But $B^\ast D = D^\ast R D $ which is self-adjoint by hypothesis on $R$ and $A^\ast C  = (\mathds{1}+C^\ast B) D^{-1} D S = S + S D^\ast R D S$, the last expression being a sum of two terms which are self-adjoint using the assumptions on $R$ and $S$.
\end{proof}

\begin{prop}
	\label{prop:Eigenvalues of Hermitian symplectic matrix are symmetric about S^1}If
	$M\in G$, then its eigenvalues are symmetric about $S^{1}$ so that its singular values are symmetric about one.
	\begin{proof}
		We have $\det( J)=1$ so that $|\det(M)|=1$.
		So define $\theta\in S^{1}$ via $\ee^{i\theta}:=\det(M^{\ast})^{-1}$.
		We have $M=- J M^{\circ} J$ from the definition so that
		if $p_{M}(\lambda)\equiv\det(M-\lambda\mathds{1}_{2N})$
		is the characteristic polynomial of $M$, we find 
		\begin{align*}
			p_{M}(\lambda) & =  \det(- J M^{\circ} J-\lambda\mathds{1}_{2N})\\
			& =  \det( M^{\circ} J-\lambda J)\tag{$\det( J)=1$}\\
			& =  e^{-i\theta}\det( J-\lambda M^{\ast} J)\tag{$e^{-i\theta}\det(M^{\ast})=1$}\\
			& =  e^{-i\theta}\det(-\mathds{1}_{2N}+\lambda M^{\ast}) =  e^{-i\theta}\overline{\det(-\mathds{1}_{2N}+\overline{\lambda}M)}\tag{$\det(- J)=1$}\\
			& =  e^{-i\theta}\overline{\overline{\lambda}^{2N}\det(-\mathds{1}_{2N}\overline{\lambda}^{-1}+M)} =  e^{-i\theta}\lambda^{2N}\overline{p_{M}(\overline{\lambda}^{-1})}\,.
		\end{align*}
		Since $|M|^2$ is also symplectic, this relation holds true also
		for $p_{|M|^2}$, hence for the singular values.
	\end{proof}
\end{prop}

\begin{cor}
	For any $M\in G$, $\|M\|=\norm{M^{-1}}$.
	\begin{proof}
		$\|M\|$ is the largest singular value, $\norm{M^{-1}}^{-1}$ is
		the smallest singular value, but the singular values are symmetric about one.
	\end{proof}
\end{cor}

\subsection{Some technical proofs necessary for \texorpdfstring{\cref{prop:fractional moments of products of transfer matrices decay exponentially}}{fractional moments of products of transfer matrices lemma}}

\begin{proof}[Proof of \cref{prop:fractional moments of products of transfer matrices decay exponentially}]
	This is essentially \cite[Prop. 2.7]{KLEIN1990135}. We spell out the proof here explicitly
	for the reader's convenience because in the reference it is merely
	outlined. 
	
	We define the additive co-cycle $\zeta$ on $\bar{L}_{j-1}\times \bar{L}_{j}$,
	which is an $\HSG$-space, via the formula 
	\begin{align*}
	\zeta(g,[y],[x]) & :=  \log(\frac{\norm{\wedge^{j-1}gy}}{\|y\|}\frac{\|x\|}{\norm{\wedge^{j}gx}})\qquad\,\forall\,(g,[y],[x])\in \HSG\times \bar{L}_{j-1}\times \bar{L}_{j}\,.
	\end{align*}
	Indeed, we have 
	\begin{align*}
	\zeta(gh,[y],[x]) & =  \zeta(g,h[y],h[x])+\zeta(h,[y],[x])\,.
	\end{align*}
	Now, we have
	\begin{align*}
	\lim_{n\to\infty}\frac{1}{n}\mathbb{E}[\zeta(B_{n}(z),[y],[x])] & =  \lim_{n\to\infty}\frac{1}{n}\mathbb{E}[\log(\frac{\norm{\wedge^{j-1}gy}}{\|y\|}\frac{\|x\|}{\norm{\wedge^{j}gx}})]\\
	& =  \lim_{n\to\infty}\frac{1}{n}\mathbb{E}[\log(\frac{\norm{\wedge^{j-1}gy}}{\|y\|})]-\lim_{n\to\infty}\frac{1}{n}\mathbb{E}[\log(\frac{\norm{\wedge^{j}gx}}{\|x\|})]\,,
	\end{align*}
	and then by \cref{prop:Convergence of Lyapunov limit uniformly}, $\frac{1}{n}\mathbb{E}[\zeta(B_{n}(z),[y],[x])]\stackrel{n\to\infty}{\longrightarrow}-\gamma_{j}(z)$
	uniformly in $z\in K$, $([x],[y])\in \bar{L}_{j}\times \bar{L}_{j-1}$.
	As a result, for any $\varepsilon>0$ there is some $n_{K}(\varepsilon)$
	such that $\forall n\in\mathbb{N}_{\geq n_{K}(\varepsilon)}$
	we have 
	\begin{align*}
	|\frac{1}{n}\mathbb{E}[\zeta(B_{n}(z),[y],[x])]+\gamma_{j}(z)| & <  \varepsilon\\
	& \downarrow\\
	\mathbb{E}[\zeta(B_{n}(z),[y],[x])] & <  n(-\gamma_{j}(z)+\varepsilon)\\
	& \downarrow\\
	\mathbb{E}[\zeta(B_{n}(z),[y],[x])] & <  n(-\Gamma_{j}(K)+\varepsilon)\,,
	\end{align*}
	so if we pick $n\geq n_{K}(\frac{1}{2}\Gamma_{j}(K))$,
	we find $\mathbb{E}[\zeta(B_{n}(z),[y],[x])]<-\frac{1}{2}n\Gamma_{j}(K)$,
	with $\Gamma_{j}(K):=\inf_{z\in K}\gamma_{j}(z)>0$
	by hypothesis.
	
	Before proceeding with the proof of \cref{prop:fractional moments of products of transfer matrices decay exponentially}, we give the following intermediate
	\begin{prop}
		There is some $n_{0}(K)$ and some $s_{K}\in(0,1)$
		such that 
		\begin{align}
		\mathbb{E}[\exp(s_{K}\zeta(B_{n_{0}(K)}(z),[y'],[x']))] & <  1-\varepsilon_{K}\label{eq:large enough product of transfer matrices is smaller than 1}
		\end{align}
		for all $y'$, $x'$ and $z\in K$, for some $\varepsilon_{K}\in(0,1)$.
		\begin{proof}
			This is the strip-analog of \cite[Lemma 5.1]{Carmona1987}. First note
			that for all $\alpha\in\mathbb{R}$ we have $e^{\alpha}  \leq  1+\alpha+\alpha^{2}e^{|\alpha|}$
			so that 
			\begin{align*}
			\exp(s\zeta(B_{n}(z),[y'],[x'])) & \leq  1+s\zeta(B_{n}(z),[y'],[x'])+s^{2}\zeta(B_{n}(z),[y'],[x'])^{2}e^{|s\zeta(B_{n}(z),[y'],[x'])|}\,.
			\end{align*}
			Next, 
			\begin{align*}
			|\zeta(B_{n}(z),[y'],[x'])| & \equiv  |\log(\frac{\norm{\wedge^{j-1}B_{n}(z)y'}}{\norm{y'}}\frac{\norm{x'}}{\norm{\wedge^{j}B_{n}(z)x'}})|\\
			& \leq  |\log(\frac{\norm{\wedge^{j-1}B_{n}(z)y'}}{\norm{y'}})|+|\log(\frac{\norm{\wedge^{j}B_{n}(z)x'}}{\norm{x'}})|\,,
			\end{align*}
			and via \cref{prop:bound on log of norm of exterior product of matrix of |det|=00003D1}
			we have 
			\begin{align*}
			|\zeta(B_{n}(z),[y'],[x'])| & \leq  (j-1)(2N-1)\log(\norm{B_{n}(z)})+j(2N-1)\log(\norm{B_{n}(z)})\\
			& \leq  4jN\log(\norm{B_{n}(z)})\\
			& \leq  4jN\sum_{k=1}^{n}\log(\norm{A_{k}(z)})\,.
			\end{align*}
			Hence by H\"older's inequality and \cref{prop:estimate on power of sum of logs of independent variables} right below
			we have:
			\begin{align*}
			\mathbb{E}[\exp(s\zeta(B_{n}(z),[y'],[x']))] & \leq  1+s\mathbb{E}[\zeta(B_{n}(z),[y'],[x'])]+s^{2}(\mathbb{E}[\zeta(B_{n}(z),[y'],[x'])^{4}])^{\frac{1}{2}}\times\\
			&   \times(\mathbb{E}[e^{2|s\zeta(B_{n}(z),[y'],[x'])|}])^{\frac{1}{2}}\\
			&   (\text{Independence of the variables})\\
			& \leq  1+s\mathbb{E}[\zeta(B_{n}(z),[y'],[x'])]+s^{2}16j^{2}N^{2}n^{2}(\mathbb{E}[\log(\norm{A_{1}(z)})^{4}])^{\frac{1}{2}}\times\\
			&   \times\mathbb{E}[\norm{A_{1}(z)}^{8sjN}]^{\frac{n}{2}}\,.
			\end{align*}
			We note that $\norm{A_{1}(z)}^{8sjN}$ is integrable by the proof of
			\cref{prop:log of transfer matrix is integrable} and \cref{assu:regularity of probability distributions}
			if we pick $s\in(0,1)$ sufficiently small. 
			
			Now with some choice of constants $C_{1}(K):=16j^{2}N^{2}\sup_{z\in K}\sqrt{\mathbb{E}[\log(\norm{A_{1}(z)}^{4})]}<\infty$
			and $C_{2}(K):=\sup_{z\in K}\sqrt{\mathbb{E}[\norm{A_{1}(z)}^{8sjN}]}$
			we find 
			\begin{align*}
			\mathbb{E}[\exp(s\zeta(B_{n}(z),[y'],[x']))] & \leq  1+s\mathbb{E}[\zeta(B_{n}(z),[y'],[x'])]+s^{2}n^{2}C_{1}(K)C_{2}(K)^{n}\,.
			\end{align*}
			From the above we know that
			\begin{align*}
			\mathbb{E}[\exp(s\zeta(B_{n_{K}(\frac{1}{2}\Gamma_{j}(K))}(z),[y'],[x']))] & \leq  1-s\frac{1}{2}n_{K}(\frac{1}{2}\Gamma_{j}(K))\Gamma_{j}(K)+\\
			&   +s^{2}n_{K}(\frac{1}{2}\Gamma_{j}(K))^{2}C_{1}(K)C_{2}(K)^{n_{K}(\frac{1}{2}\Gamma_{j}(K))}\,.
			\end{align*}
			Hence there is some $s\in(0,1)$ (which depends on $K$)
			so that 
			\begin{align*}
			\varepsilon_{K} & :=  \frac{1}{2}n_{K}(\frac{1}{2}\Gamma_{j}(K))\Gamma_{j}(K)-sn_{K}(\frac{1}{2}\Gamma_{j}(K))^{2}C_{1}(K)C_{2}(K)^{n_{K}(\frac{1}{2}\Gamma_{j}(K))}\,.
			\end{align*}
			is positive.
		\end{proof}
	\end{prop}
	Continuing now with the proof of \cref{prop:fractional moments of products of transfer matrices decay exponentially}, we note that the object whose expectation we're actually trying
	to bound is
	\begin{align*}
	\exp(s\zeta(B_{n}(z),[y],[x])) & \equiv  \exp(s\log(\frac{\norm{\wedge^{j-1}gy}}{\|y\|}\frac{\|x\|}{\norm{\wedge^{j}gx}}))\\
	& =  (\frac{\norm{\wedge^{j-1}gy}}{\|y\|}\frac{\|x\|}{\norm{\wedge^{j}gx}})^{s}\,,
	\end{align*}
	and that 
	\begin{align*}
	\exp(s\zeta(B_{n+m}(z),[y],[x])) & \equiv  \exp(s\zeta(A_{n+m}(z)\dots A_{1+m}(z)B_{m}(z),[y],[x]))\\
	&   \quad(\text{cocycle property})\\
	& =  \exp(s(\zeta(A_{n+m}(z)\dots A_{1+m}(z),B_{m}(z)[y],B_{m}(z)[x])+\\
	& \quad+ \zeta(B_{m}(z),[y],[x])))\\
	& =  \exp(s\zeta(A_{n+m}(z)\dots A_{1+m}(z),B_{m}(z)[y],B_{m}(z)[x]))\times\\
	& \quad\times \exp(s\zeta(B_{m}(z),[y],[x]))\,.
	\end{align*}
	Hence due to the fact that $\{A_{n}(z)\}_{n\in\mathbb{Z}}$ are
	independent, we can integrate first only over $\{A_{n+m}(z),\dots,A_{1+m}(z)\}$,
	so that in that integration $B_{m}(z)y$ and $B_{m}(z)x$
	are fixed. Then via \cref{eq:large enough product of transfer matrices is smaller than 1}
	we find 
	\begin{align*}
	\mathbb{E}[\exp(s_{K}\zeta(B_{n_{0}(K)+m}(z),[y],[x]))] & \leq  (1-\varepsilon_{K})\mathbb{E}[\exp(s_{K}\zeta(B_{m}(z),[y],[x]))]\,.
	\end{align*}
	So if $n\in\mathbb{N}_{\geq n_{0}(K)}$ is given, we
	write it as $n=q_{n}n_{0}(K)+r_{n}$ with $r_{n}\in\{0,\dots,n_{0}(K)-1\}$.
	We then have using H\"older's inequality
	\begin{align*}
	&   \mathbb{E}[\exp(\frac{1}{2}s_{K}\zeta(B_{n}(z),[y],[x]))]\\
	& \leq  \mathbb{E}\exp(\frac{1}{2}s_{K}\zeta(A_{q_{n}n_{0}(K)+r_{n}}(z)\dots A_{q_{n}n_{0}(K)+1}(z),B_{q_{n}n_{0}(K)}(z)[y],B_{q_{n}n_{0}(K)}(z)[x]))\times\\
	&   \times\exp(\frac{1}{2}s_{K}\zeta(B_{q_{n}n_{0}(K)}(z),[y],[x]))\\
	& \leq  \sqrt{\mathbb{E}[\exp(s_{K}\zeta(A_{q_{n}n_{0}(K)+r_{n}}(z)\dots A_{q_{n}n_{0}(K)+1}(z),B_{q_{n}n_{0}(K)}(z)[y],B_{q_{n}n_{0}(K)}(z)[x]))]}\times\\
	&   \times\sqrt{\mathbb{E}[\exp(s_{K}\zeta(B_{q_{n}n_{0}(K)}(z),[y],[x]))]}\\
	& \leq  \sqrt{\mathbb{E}[(\norm{\wedge^{j-1}A_{q_{n}n_{0}(K)+r_{n}}(z)\dots A_{q_{n}n_{0}(K)+1}(z)}\norm{(\wedge^{j}A_{q_{n}n_{0}(K)+r_{n}}(z)\dots A_{q_{n}n_{0}(K)+1}(z))^{-1}})^{s_{K}}]}\times\\
	& \times(1-\varepsilon_{K})^{\frac{q_{n}}{2}}\,,
	\end{align*}
	and using the proof of \cref{prop:bound on log of norm of exterior product of matrix of |det|=00003D1}
	and the independence condition (assuming again that $s_{K}$ has to
	be redefined so that $\norm{A_{1}}^{(j-1)s_{K}+(2N-1)js_{K}}$
	is also integrable (via \cref{assu:regularity of probability distributions})) we find
	\begin{align*}
	\mathbb{E}[\exp(\frac{1}{2}s_{K}\zeta(B_{n}(z),[y],[x]))] & \leq  \mathbb{E}[\norm{A_{1}(z)}^{(j-1)s_{K}+(2N-1)js_{K}}]^{\frac{r_{n}}{2}}(1-\varepsilon_{K})^{\frac{q_{n}}{2}}\\
	& \leq  (\sup_{z\in K}\mathbb{E}[\norm{A_{1}(z)}^{(j-1)s_{K}+(2N-1)js_{K}}]^{n_{0}(K)})(1-\varepsilon_{K})^{\frac{q_{n}}{2}}\,,
	\end{align*}
	which implies the bound in the claim.
\end{proof}

In the proof above we have used the following two propositions. The first one is a basic consequence of the multinomial theorem, Jensen's inequality, and independence and hence its proof is omitted.
\begin{prop}
\label{prop:estimate on power of sum of logs of independent variables}We
have 
\begin{align*}
\mathbb{E}[(\sum_{j=1}^{n}\log(\norm{A_{j}}))^{4}] & \leq  n^{4}\mathbb{E}[\log(\norm{A_{1}})^{4}]
\end{align*}
where $\{A_{j}\}_{j}$ are independent variables.
\end{prop}

\begin{prop}
\label{prop:bound on log of norm of exterior product of matrix of |det|=00003D1}(\cite{Bougerol9781468491746}
Lemma III.5.4) If $M\in \Mat_{L\times L}(\mathbb{C})$ has
$|\det(M)|=1$ then for any $j\in\{1,\dots,L\}$
and $v\in\wedge^{j}\mathbb{C}^{L}$ we have 
\begin{align*}
|\log(\norm{\wedge^{j}M})| & \leq  j(L-1)\log(\|M\|)
\end{align*}
and
\begin{align*}
|\log(\frac{\norm{\wedge^{j}Mv}}{\|v\|})| & \leq  j(L-1)\log(\|M\|)\,.
\end{align*}
\begin{proof}
First note that $\norm{\wedge^{j}M}$ is the product of the first
$j$ singular values: 
\begin{align*}
\norm{\wedge^{j}M} & =  \sigma_{1}(M)\dots\sigma_{j}(M)
\end{align*}
so that 
\begin{align*}
\norm{\wedge^{j}M} & \leq  \|M\|^{j}\,.
\end{align*}
Conversely, 
\begin{align*}
\norm{\wedge^{j}M}^{-1} & \leq  \norm{(\wedge^{j}M)^{-1}} =  \norm{\wedge^{j}M^{-1}} \leq  \norm{M^{-1}}^{j}\,.
\end{align*}
For any invertible matrix we have 
\begin{align*}
\norm{M^{-1}} & \leq  \frac{\|M\|^{L-1}}{|\det(M)|}\,,
\end{align*}
so that in our case 
\begin{align*}
\norm{\wedge^{j}M}^{-1} & \leq  \|M\|^{j(L-1)}\,.
\end{align*}
We find using the fact that $\log$ is monotone increasing,
\begin{align*}
\log(\norm{\wedge^{j}M}) & \leq  j\log(\|M\|) \leq  j(L-1)\log(\|M\|)
\end{align*}
and 
\begin{align*}
-\log(\norm{\wedge^{j}M}) & =  \log(\norm{\wedge^{j}M}^{-1}) \leq  j(L-1)\log(\|M\|)\,.
\end{align*}
Next, we have 
\begin{align*}
\frac{\norm{\wedge^{j}Mv}}{\|v\|} & \leq  \norm{\wedge^{j}M}\,,
\end{align*}
whereas $\|v\|=\norm{(\wedge^{j}M)^{-1}(\wedge^{j}M)v}\leq\norm{(\wedge^{j}M)^{-1}}\norm{\wedge^{j}Mv}$
so that 
\begin{align*}
(\frac{\norm{\wedge^{j}Mv}}{\|v\|})^{-1} & \leq  \norm{(\wedge^{j}M)^{-1}}
\end{align*}
which gives the second inequality of the prop.

Finally, note that these inequalities indeed make sense: $1\leq\|M\|\norm{M^{-1}}\leq\|M\|\|M\|^{L-1}=\|M\|^{L}$
so that $\log(\|M\|)\geq0$ always when $|\det(M)|=1$.
\end{proof}
\end{prop}
%

\begingroup
\let\itshape\upshape
 \printbibliography

@book{Bredon_1993,
  Author = {Glen E. Bredon},
  Title = {Topology and Geometry (Graduate Texts in Mathematics)},
  Publisher = {Springer},
  Year = {1993},
  ISBN = {},
  URL = {https://www.amazon.com/Topology-Geometry-Graduate-Texts-Mathematics-ebook/dp/B000WLYYKY?SubscriptionId=0JYN1NVW651KCA56C102&tag=techkie-20&linkCode=xm2&camp=2025&creative=165953&creativeASIN=B000WLYYKY}
}

@ARTICLE{Combes_Thomas_1973,
	author = {{Combes}, J.~M. and {Thomas}, L.},
	title = "{Asymptotic behaviour of eigenfunctions for multiparticle Schr{\"o}dinger operators}",
	journal = {Commun. Math. Phys.},
	year = 1973,
	month = dec,
	volume = 34,
	pages = {251-270},
	doi = {10.1007/BF01646473},
	adsurl = {http://adsabs.harvard.edu/abs/1973CMaPh..34..251C}
}

@Article{Aizenman2001,
	author="Aizenman, Michael
	and Schenker, Jeffrey H.
	and Friedrich, Roland M.
	and Hundertmark, Dirk",
	title="Finite-volume fractional-moment criteria for {Anderson} localization",
	journal="Commun. Math. Phys.",
	year="2001",
	month="Nov",
	day="01",
	volume="224",
	number="1",
	pages="219--253",
	issn="1432-0916",
	doi="10.1007/s002200100441",
	url="https://doi.org/10.1007/s002200100441"
}

@Article{Carmona1987,
	author="Carmona, Rene
	and Klein, Abel
	and Martinelli, Fabio",
	title="Anderson localization for {Bernoulli} and other singular potentials",
	journal="Commun. Math. Phys.",
	year="1987",
	volume="108",
	number="1",
	pages="41--66",
	issn="1432-0916",
	doi="10.1007/BF01210702",
	url="http://dx.doi.org/10.1007/BF01210702"
}

@Book{Carmona1990,
	author="Carmona, Ren{\'e}
	and Lacroix, Jean",
	title="Spectral Theory of Random Schr{\"o}dinger Operators",
	year="1990",
	publisher="Birkh{\"a}user Boston",
%	address="Boston, MA",
%	pages="175--240",
	isbn="978-1-4612-4488-2",
	doi="10.1007/978-1-4612-4488-2_4",
	url="http://dx.doi.org/10.1007/978-1-4612-4488-2_4"
}

@Article{Simon1985,
author="Simon, Barry",
title="Localization in general one dimensional random systems, {I}. {J}acobi matrices",
journal="Commun. Math. Phys.",
year="1985",
volume="102",
number="2",
pages="327--336",
issn="1432-0916",
doi="10.1007/BF01229383",
url="http://dx.doi.org/10.1007/BF01229383"
}

@book{Bougerol9781468491746,
  Author = {P. Bougerol and Lacroix},
  Title = {Products of Random Matrices with Applications to Schr{\"o}dinger Operators (Progress in Probability)},
  Publisher = {Birkh{\"a}user},
  Year = {2012},
  ISBN = {1468491741},
  URL = {https://www.amazon.com/Products-Applications-Schr%C3%B6dinger-Operators-Probability/dp/1468491741?SubscriptionId=0JYN1NVW651KCA56C102&tag=techkie-20&linkCode=xm2&camp=2025&creative=165953&creativeASIN=1468491741}
}

@ARTICLE{Aizenman_Graf_1998,
   author = {{Aizenman}, M. and {Graf}, G.~M.},
    title = "{Localization bounds for an electron gas}",
  journal = {J. Phys. A Math. Gen.},
   eprint = {cond-mat/9603116},
     year = 1998,
    month = aug,
   volume = 31,
    pages = {6783-6806},
      doi = {10.1088/0305-4470/31/32/004},
   adsurl = {http://adsabs.harvard.edu/abs/1998JPhA...31.6783A}
}

@Article{Chapman2015,
author="Chapman, Jacob
and Stolz, G{\"u}nter",
title="Localization for random block operators related to the XY spin chain",
journal="Ann. Inst. Henri Poincar{\'e}",
year="2015",
volume="16",
number="2",
pages="405--435",
issn="1424-0661",
doi="10.1007/s00023-014-0328-2",
url="http://dx.doi.org/10.1007/s00023-014-0328-2"
}

@Article{Jitomirskaya2003,
author="Jitomirskaya, S.
and Schulz-Baldes, H.
and Stolz, G.",
title="Delocalization in random polymer models",
journal="Commun. Math. Phys.",
year="2003",
month="Feb",
day="01",
volume="233",
number="1",
pages="27--48",
issn="1432-0916",
doi="10.1007/s00220-002-0757-5",
url="https://doi.org/10.1007/s00220-002-0757-5"
}

@article{KLEIN1990135,
title = "Localization for the Anderson model on a strip with singular potentials",
journal = "J. Funct. Anal.",
volume = "94",
number = "1",
pages = "135 - 155",
year = "1990",
note = "",
issn = "0022-1236",
doi = "http://dx.doi.org/10.1016/0022-1236(90)90031-F",
url = "http://www.sciencedirect.com/science/article/pii/002212369090031F",
author = "Abel Klein and Jean Lacroix and Athanasios Speis",

}

@article{SSH_1979,
  title = {Solitons in polyacetylene},
  author = {Su, W. P. and Schrieffer, J. R. and Heeger, A. J.},
  journal = {Phys. Rev. Lett.},
  volume = {42},
  issue = {25},
  pages = {1698--1701},
  numpages = {0},
  year = {1979},
  month = {6},
  publisher = {American Physical Society},
  doi = {10.1103/PhysRevLett.42.1698},
  url = {http://link.aps.org/doi/10.1103/PhysRevLett.42.1698}
}

@Article{Kunz1980,
author="Kunz, Herv{\'e}
and Souillard, Bernard",
title="Sur le spectre des op{\'e}rateurs aux diff{\'e}rences finies al{\'e}atoires",
journal="Commun. Math. Phys.",
year="1980",
month="12",
volume="78",
number="2",
pages="201--246",
issn="1432-0916",
doi="10.1007/BF01942371",
url="https://doi.org/10.1007/BF01942371"
}

@ARTICLE{Graf_Shapiro_2018_1D_Chiral_BEC,
author="{Graf}, G.~M..
and Shapiro, J.",
title="The bulk-edge correspondence for disordered chiral chains",
journal="Commun. Math. Phys.",
year="2018",
volume="363",
number="3",
pages="829--846",
}

@Article{EGS_2005,
author="Elgart, A.
and {Graf}, G.~M..
and Schenker, J.H.",
title="Equality of the bulk and edge {H}all conductances in a mobility gap",
journal="Commun. Math. Phys.",
year="2005",
volume="259",
number="1",
pages="185--221",
issn="1432-0916",
doi="10.1007/s00220-005-1369-7",
url="http://dx.doi.org/10.1007/s00220-005-1369-7"
}

@Article{Ludwig2013,
	author="Ludwig, Andreas W. W.
	and Schulz-Baldes, Hermann
	and Stolz, Michael",
	title="Lyapunov Spectra for All Ten Symmetry Classes of Quasi-one-dimensional Disordered Systems of Non-interacting Fermions",
	journal="Journal of Statistical Physics",
	year="2013",
	month="Jul",
	day="01",
	volume="152",
	number="2",
	pages="275--304",
	issn="1572-9613",
	doi="10.1007/s10955-013-0764-2",
	url="https://doi.org/10.1007/s10955-013-0764-2"
}

@Article{Asch2010,
	author="Asch, Joachim
	and Bourget, Olivier
	and Joye, Alain",
	title="Localization Properties of the Chalker--Coddington Model",
	journal="Annales Henri Poincar{\'e}",
	year="2010",
	month="Dec",
	day="01",
	volume="11",
	number="7",
	pages="1341--1373",
	issn="1424-0661",
	doi="10.1007/s00023-010-0056-1",
	url="https://doi.org/10.1007/s00023-010-0056-1"
}

@Article{Sadel2010,
	author="Sadel, Christian
	and Schulz-Baldes, Hermann",
	title="Random Dirac Operators with Time Reversal Symmetry",
	journal="Communications in Mathematical Physics",
	year="2010",
	month="Apr",
	day="01",
	volume="295",
	number="1",
	pages="209--242",
	issn="1432-0916",
	doi="10.1007/s00220-009-0956-4",
	url="https://doi.org/10.1007/s00220-009-0956-4"
}

@Book{attal2006open,
	author = {Attal, S. and Joye, A. and Pillet, C.-A. },
	title = {Open Quantum Systems I. The Hamiltonian Approach},
	publisher = {Springer},
	year = {2006},
	address = {Berlin},
	isbn = {978-3-540-33922-9}
}

@incollection {MR2509108,
	AUTHOR = {Hislop, Peter D.},
	TITLE = {Lectures on random {S}chr\"{o}dinger operators},
	BOOKTITLE = {Fourth {S}ummer {S}chool in {A}nalysis and {M}athematical
	{P}hysics},
	SERIES = {Contemp. Math.},
	VOLUME = {476},
	PAGES = {41--131},
	PUBLISHER = {Amer. Math. Soc., Providence, RI},
	YEAR = {2008},
	MRCLASS = {82B44 (35J10 35P05 47B80 47N50 60H25 81Q10)},
	MRNUMBER = {2509108},
	MRREVIEWER = {J. A. Van Casteren},
	DOI = {10.1090/conm/476/09293},
	URL = {https://doi.org/10.1090/conm/476/09293},
}

@TECHREPORT{Kirsch07aninvitation,
	author = {Werner Kirsch},
	title = {An invitation to random {S}chr\"odinger operators},
	institution = {},
	year = {2007}
}

@article{Hasan_Kane_2010,
	title = {Colloquium: Topological insulators},
	author = {Hasan, M. Z. and Kane, C. L.},
	journal = {Rev. Mod. Phys.},
	volume = {82},
	issue = {4},
	pages = {3045--3067},
	numpages = {0},
	year = {2010},
	month = {11},
	publisher = {American Physical Society},
	doi = {10.1103/RevModPhys.82.3045},
	url = {https://link.aps.org/doi/10.1103/RevModPhys.82.3045}
}

@article{AltlandZirnbauer97,
	title = {Nonstandard symmetry classes in mesoscopic normal-superconducting hybrid structures},
	author = {Altland, Alexander and Zirnbauer, Martin R.},
	journal = {Phys. Rev. B},
	volume = {55},
	issue = {2},
	pages = {1142--1161},
	numpages = {0},
	year = {1997},
	month = {01},
	publisher = {American Physical Society},
	doi = {10.1103/PhysRevB.55.1142},
	url = {https://link.aps.org/doi/10.1103/PhysRevB.55.1142}
}

@Article{Damanik2001,
	author={Damanik, D.
	and Stollmann, P.},
	title={Multi-scale analysis implies strong dynamical localization},
	journal={Geometric {\&} Functional Analysis GAFA},
	year={2001},
	month={Apr},
	day={01},
	volume={11},
	number={1},
	pages={11-29},
	issn={1420-8970},
	doi={10.1007/PL00001666},
	url={https://doi.org/10.1007/PL00001666}
}

@Article{Graf1994,
	author={Graf, Gian Michele},
	title={Anderson localization and the space-time characteristic of continuum states},
	journal={Journal of Statistical Physics},
	year={1994},
	month={Apr},
	day={01},
	volume={75},
	number={1},
	pages={337-346},
	issn={1572-9613},
	doi={10.1007/BF02186292},
	url={https://doi.org/10.1007/BF02186292}
}

@Article{Shapiro2019,
	author={Shapiro, Jacob
	and Tauber, Cl{\'e}ment},
	title={Strongly Disordered Floquet Topological Systems},
	journal={Annales Henri Poincar{\'e}},
	year={2019},
	month={Jun},
	day={01},
	volume={20},
	number={6},
	pages={1837-1875},
	issn={1424-0661},
	doi={10.1007/s00023-019-00794-3},
	url={https://doi.org/10.1007/s00023-019-00794-3}
}

@article{PhysRevLett.64.1851,
	title = {Scaling properties of band random matrices},
	author = {Casati, Giulio and Molinari, Luca and Izrailev, Felix},
	journal = {Phys. Rev. Lett.},
	volume = {64},
	issue = {16},
	pages = {1851--1854},
	numpages = {0},
	year = {1990},
	month = {Apr},
	publisher = {American Physical Society},
	doi = {10.1103/PhysRevLett.64.1851},
	url = {https://link.aps.org/doi/10.1103/PhysRevLett.64.1851}
}

@article{PhysRevE.48.R1613,
	title = {Band-random-matrix model for quantum localization in conservative systems},
	author = {Casati, G. and Chirikov, B. V. and Guarneri, I. and Izrailev, F. M.},
	journal = {Phys. Rev. E},
	volume = {48},
	issue = {3},
	pages = {R1613--R1616},
	numpages = {0},
	year = {1993},
	month = {Sep},
	publisher = {American Physical Society},
	doi = {10.1103/PhysRevE.48.R1613},
	url = {https://link.aps.org/doi/10.1103/PhysRevE.48.R1613}
}

@Article{Schenker2009,
	author={Schenker, Jeffrey},
	title={Eigenvector Localization for Random Band Matrices with Power Law Band Width},
	journal={Communications in Mathematical Physics},
	year={2009},
	month={Sep},
	day={01},
	volume={290},
	number={3},
	pages={1065-1097},
	issn={1432-0916},
	doi={10.1007/s00220-009-0798-0},
	url={https://doi.org/10.1007/s00220-009-0798-0}
}

@article{10.1093/imrn/rnx145,
	author = {Peled, Ron and Schenker, Jeffrey and Shamis, Mira and Sodin, Sasha},
	title = "{On the Wegner Orbital Model}",
	journal = {International Mathematics Research Notices},
	volume = {2019},
	number = {4},
	pages = {1030-1058},
	year = {2017},
	month = {07},
	issn = {1073-7928},
	doi = {10.1093/imrn/rnx145},
	url = {https://doi.org/10.1093/imrn/rnx145},
	eprint = {https://academic.oup.com/imrn/article-pdf/2019/4/1030/27844186/rnx145.pdf},
}

@article{Bourgade2017,
	doi = {10.4310/atmp.2017.v21.n3.a5},
	url = {https://doi.org/10.4310/atmp.2017.v21.n3.a5},
	year = {2017},
	publisher = {International Press of Boston},
	volume = {21},
	number = {3},
	pages = {739--800},
	author = {Paul Bourgade and Laszlo Erd{\H{o}}s and Horng-Tzer Yau and Jun Yin},
	title = {Universality for a class of random band matrices},
	journal = {Advances in Theoretical and Mathematical Physics}
}

@article{https://doi.org/10.1002/cpa.21895,
	author = {Bourgade, Paul and Yau, Horng-Tzer and Yin, Jun},
	title = {Random Band Matrices in the Delocalized Phase I: Quantum Unique Ergodicity and Universality},
	journal = {Communications on Pure and Applied Mathematics},
	volume = {73},
	number = {7},
	pages = {1526-1596},
	doi = {https://doi.org/10.1002/cpa.21895},
	url = {https://onlinelibrary.wiley.com/doi/abs/10.1002/cpa.21895},
	eprint = {https://onlinelibrary.wiley.com/doi/pdf/10.1002/cpa.21895},
	year = {2020}
}

@Article{Shcherbina2014,
	author={Shcherbina, Tatyana},
	title={On the Second Mixed Moment of the Characteristic Polynomials of 1D Band Matrices},
	journal={Communications in Mathematical Physics},
	year={2014},
	month={May},
	day={01},
	volume={328},
	number={1},
	pages={45-82},
	issn={1432-0916},
	doi={10.1007/s00220-014-1947-7},
	url={https://doi.org/10.1007/s00220-014-1947-7}
}

@Article{Shcherbina2017,
	author={Shcherbina, Mariya
	and Shcherbina, Tatyana},
	title={Characteristic Polynomials for 1D Random Band Matrices from the Localization Side},
	journal={Communications in Mathematical Physics},
	year={2017},
	month={May},
	day={01},
	volume={351},
	number={3},
	pages={1009-1044},
	issn={1432-0916},
	doi={10.1007/s00220-017-2849-2},
	url={https://doi.org/10.1007/s00220-017-2849-2}
}

@article{Newman_1986_cmp/1104114627,
	author = {Charles M. Newman},
	title = {{The distribution of Lyapunov exponents: exact results for random matrices}},
	volume = {103},
	journal = {Communications in Mathematical Physics},
	number = {1},
	publisher = {Springer},
	pages = {121 -- 126},
	year = {1986},
	doi = {cmp/1104114627},
	URL = {https://doi.org/}
}
 \endgroup
\end{document}